\documentclass[aps,pra,reprint,superscriptaddress]{revtex4-2}
\usepackage{graphicx}
\usepackage{braket}
\usepackage{amsmath}
\usepackage{amsthm}
\usepackage{amssymb}
\usepackage{dsfont}
\usepackage[colorlinks,
            linkcolor=blue,
            anchorcolor=black,
            citecolor=blue,
			urlcolor=blue]{hyperref}
			
\usepackage[ruled]{algorithm2e}
\usepackage{makecell}
\usepackage{array}
\usepackage{color}
\usepackage{multirow}
\newtheorem{theorem}{Theorem}

\newtheorem{definition}{Definition}
\newtheorem{corollary}{Corollary}
\newtheorem{lemma}{Lemma}

\begin{document}


\title{High-dimensional quantum XYZ product codes for biased noise}


\author{Zhipeng Liang}
\affiliation{Harbin Institute of Technology, Shenzhen. Shenzhen, 518055, China}
\author{Zhengzhong Yi}
\email[]{zhengzhongyi@cs.hitsz.edu.cn} 
\affiliation{Hefei National Research Center for Physical Sciences at the Microscale and School of Physical Sciences, University of Science and Technology of China, Hefei 230026, China}
\affiliation{Shanghai Research Center for Quantum Science and CAS Center for Excellence in Quantum Information and Quantum Physics, University of Science and Technology of China, Shanghai 201315, China}
\affiliation{Hefei National Laboratory, University of Science and Technology of China, Hefei 230088, China}
\author{Fusheng Yang}
\affiliation{Harbin Institute of Technology, Shenzhen. Shenzhen, 518055, China}
\author{Jiahan Chen}
\affiliation{Harbin Institute of Technology, Shenzhen. Shenzhen, 518055, China}
\author{Zicheng Wang}
\affiliation{Harbin Institute of Technology, Shenzhen. Shenzhen, 518055, China}
\author{Xuan Wang}
\email[]{wangxuan@cs.hitsz.edu.cn}
\affiliation{Harbin Institute of Technology, Shenzhen. Shenzhen, 518055, China}


\date{\today}

\begin{abstract}
Three-dimensional (3D) quantum XYZ product can construct a class of non-CSS quantum codes by using three classical codes. However, there has been limited study on their error-correcting performance so far and whether this code construction can be generalized to higher dimension is an open question. In this paper, we first study the error-correcting performance of the 3D Chamon code, which is an instance of the 3D XYZ product of three repetition codes. Second, we show that the 3D XYZ product can be generalized to four dimension and propose four-dimensional (4D) XYZ product code construction, which constructs a class of non-CSS quantum codes by using either four classical codes or two CSS quantum codes. Compared with the 4D homological product, we show that the 4D XYZ product can construct non-CSS codes with higher code dimension or code distance. Third, we consider two instances of the 4D XYZ product, to which we refer as the 4D Chamon code and the 4D XYZ product concatenated code, respectively. Our simulation results show that, the 4D XYZ product can construct non-CSS codes with better error-correcting performance against Pauli-$Z$-biased noise than CSS codes constructed by the 4D homological product. Finally, we present the geometric arrangement of the 4D Chamon code within a 4D cubic lattice, demonstrating that it possesses two key characteristics of fracton models, which strongly suggest that it is a novel 4D fracton model. 
\end{abstract}


\maketitle

\section{Introduction}
\label{1}

Quantum error-correcting code (QECC) \cite{PhysRevA.52.R2493,PhysRevLett.77.793} provides a promising approach to realize large-scale fault-tolerant quantum computing. In recent years, quantum low density parity-check (QLDPC) codes \cite{breuckmann2021quantum,10.5555/2685179.2685184} have attracted a great deal of attention, since they may have high coding rate, large code distance and low stabilizer weight, which benefits their engineering application. So far, many constructions of QLDPC codes have been proposed, such as XYZ product \cite{leverrier2022quantum}, hypergraph product \cite{tillich2013quantum}, homological product \cite{10.1145/2591796.2591870}, lifted product \cite{9567703}, and balanced product \cite{9490244}. All the above construction methods generate CSS codes \cite{681315}, except XYZ product which generates non-CSS codes.

Compared with CSS codes, non-CSS codes have shown better error-correcting performance against biased noise \cite{bonilla2021xzzx,PhysRevResearch.5.013035,6qbg-xslr}. In quantum computing hardware system, the probability of Pauli $Z$ noise is often much higher than the probabilities of Pauli $X$ and $Y$ noise \cite{Aliferis_2009,Demonstration,topologically2014}, thus considering quantum error-correcting codes with better performance against biased noise is more meaningful in this sense.

Three-dimensional (3D) XYZ product proposed in Ref. \cite{leverrier2022quantum} constructs a class of non-CSS codes by using three classical codes. However, so far there is limited study on the error-correcting performance of 3D XYZ product codes, and there exists some natural questions: whether this construction method can be generalized to higher dimension? If possible, how is the error-correcting performance of these high-dimensional XYZ product codes, especially against biased noise?

The motivation of our work is to answer the problems above. In this paper, we first study the error-correcting performance of the 3D Chamon code \cite{PhysRevLett.94.040402,BRAVYI2011839}, which is an instance of the 3D XYZ product of three repetition codes with length $n_1$, $n_2$, $n_3$. Exploiting fully decoupled binary belief propagation combined with $order\mbox{-}0$ ordered statistics decoding (FDBP-OSD-0 \cite{yi2025improved}), when $n_1=n_2=n_3$, under the depolarizing noise and pure Pauli $X$, $Y$ and $Z$ noise, the code-capacity thresholds of the Chamon code are all around $14.5\%$. 

Second, we show that that the 3D XYZ product can be generalized to four dimension and propose the four-dimensional (4D) XYZ product code construction, which constructs a non-CSS code by making use of either four classical codes or two CSS codes. It also can be regarded as a non-CSS variant of the standard homological product (namely, 4D homological product) of two CSS codes .  

Third, we show that using two identical component CSS codes, 4D XYZ product can construct codes with higher code dimension (the number of encoded logical qubits) or code distance than those constructed by the 4D homological product.

Fourth, to explore the error-correcting performance of 4D XYZ product codes, we consider two instances. The first one is constructed from two quantum concatenated codes, each of which is derived from a pair of repetition codes. The second one is constructed from two hypergraph product codes, similarly obtained from two pairs of repetition codes. We refer to the codes as the 4D XYZ product concatenated code and the 4D Chamon code, respectively. Exploiting FDBP-OSD-0, our simulation results show that under the pure Pauli $Z$ noise, the code-capacity thresholds of both the 4D XYZ product concatenated code and the 4D Chamon code are much higher than those of the 4D toric code \cite{breuckmann2016local} and 4D homological product concatenated code (namely, the codes constructed by 4D homological product of two quantum concatenated codes, which are the two identical component codes used to construct 4D XYZ product concatenated code) respectively, while their code-capacity thresholds are close under depolarizing noise. Furthermore, under Pauli-$Z$-biased noise with varying bias rates, the 4D XYZ product concatenated code and the 4D Chamon code both exhibit strong error-correction performance in terms of code-capacity threshold. These results indicate that, when utilizing identical component codes, the 4D XYZ product construction yields non-CSS codes with superior performance against biased noise compared to CSS codes generated by the 4D homological product.

At last, we study the geometric arrangement of the 4D Chamon code within a 4D cubic lattice and demonstrate that it possesses two key characteristics of fracton models—ground state degeneracy dependent on cubic lattice size and the restricted mobility of excitations, which strongly suggest that it is a novel 4D fracton model. Our simulation results also demonstrate that the error-correcting performance of the 4D Chamon code is significantly better than that of the other fracton model-the 4D X-cube model, which implies that the 4D Chamon code is a promising candidate for self-correcting quantum memory.

This paper is organized as follow. Sect. \ref{2} introduces some preliminaries, including quantum stabilizer codes, chain complex, hypergraph product and 3D XYZ product code construction. In Sect. \ref{3}, we introduce the 4D XYZ product code construction and study its code dimension and code distance. In Sect. \ref{4}, we explore the error-correcting performance of 3D Chamon codes, 4D Chamon codes and 4D XYZ product concatenated codes. In Sect. \ref{Further discussions}, we provide more discussions on the 4D Chamon code. In Sect. \ref{6}, we conclude our work.

\section {Preliminaries}
\label{2}

\subsection {Quantum stabilizer codes}
\label{QSCs}
Quantum stabilizer codes (QSCs) \cite{gottesman1997stabilizer} are an important class of quantum error correcting codes, which can be seen as the analogue of classical linear codes in the quantum field.

Giving a $[[n,k]]$ QSC $C$ is equivalent to explicitly giving a set of independent $n$-qubits Pauli operators $S_1,\cdots,S_{(n-k)}\in \mathcal{P}_n$, which commute with each other and are called stabilizer generators, where $\mathcal{P}_n$ is the n-qubits Pauli group. This is because the code space $\mathcal{P}_C$ of $C$ is the common +1 eigenspace of $S_1,\cdots,S_{(n-k)}\in \mathcal{P}_n$. Formally,
\begin{equation}
	\mathcal{Q}_C=\left\{\ket{\varphi}\in\mathcal{H}_2^{\left(\otimes n\right)}:S\ket{\varphi}=\ket{\varphi},\forall S\in\mathcal{S}\right\}
\end{equation}
where $\mathcal{H}_2^{\left(\otimes n\right)}$ is the $n$-qubits Hilbert space and $\mathcal{S}$ is the stabilizer group, which is generated by $S_1,\cdots,S_{\left(n-k\right)}$, namely $\mathcal{S}=\left\langle S_1,\cdots,S_{\left(n-k\right)}\right\rangle$.

If the stabilizer generators of a QSC $C$ can be divided into two parts, each of which only contains either Pauli $X$ or Pauli $Z$ operators, it is a CSS code, otherwise it is a non-CSS code. It can be seen that each stabilizer generator of a CSS code can only detect two types of Pauli errors, while that of a non-CSS code can detect all three types of Pauli errors.

The logical operators of a QSC $C$ are the set of operators in $\mathcal{P}_n$ which commute with all elements in $\mathcal{S}$ but are not in $\mathcal{S}$. Formally, the logical operators are the elements of $\mathcal{C}\left(\mathcal{S}\right)\backslash \mathcal{S}$, where $\mathcal{C}\left(\mathcal{S}\right)$ is the centralizer of $\mathcal{S}$, which is defined as $\mathcal{C}\left(\mathcal{S}\right)=\left\{P\in\mathcal{P}_n:SP=PS,\forall S\in\mathcal{S}\right\}$. For a $\left[\left[n,k,d\right]\right]$ QSC, we can find $k$ pairs of logical operators $\left({\bar{X}}_j,{\bar{Z}}_j\right)_{j=1,\cdots,k}$ such that ${\bar{X}}_i{\bar{Z}}_j=\left(-1\right)^{\delta_{ij}}{\bar{Z}}_j{\bar{X}}_i$, where $\delta$ is the Kronecker delta, which means for the same pair of logical operators ${\bar{X}}_j,{\bar{Z}}_j$, they are anti-commute, but they commute with other pairs of logical operators. We can see that $\mathcal{C}\left(\mathcal{S}\right)=\left\langle S_1,\cdots,S_{\left(n-k\right)},{\bar{X}}_1,{\bar{Z}}_1,\cdots,{\bar{X}}_k,{\bar{Z}}_k\right\rangle$. The code distance $d$ is defined as the minimum weight of logical operators, namely,
\begin{equation}
	d=\mathop{\min}_{L\in\mathcal{C}\left(\mathcal{S}\right)\backslash \mathcal{S}}wt\left(L\right)
\end{equation}
where the weight of an operator $P\in\mathcal{P}_n$ is defined as the number of qubits on which it acts nontrivially, and we use notation $wt\left(\cdot\right)$ to denote it. For instance, $wt\left(I_1X_2Y_3Z_4\right)=3$.

\subsection {Chain complex}
\label{chain complex}
In this section, we introduce the concept of chain complex, which will help to better understand the hypergraph product code, 3D XYZ product code and our proposed 4D XYZ product code.

A chain complex $\mathfrak{C}$ with length $L$ is a collection of $L+1$ vector spaces $C_0,\ C_1,\cdots,C_L$ and $L$ linear maps (which are also called boundary operators) $\partial_i:C_i\rightarrow C_{i+1}\ \left(0\le i\le L-1\right)$, namely,
\begin{equation}
	\mathfrak{C}=C_0\stackrel{\partial_0}{\longrightarrow}C_1\stackrel{\partial_1}{\longrightarrow}\cdots\stackrel{\partial_{L-1}}{\longrightarrow}C_L
\end{equation}
which satisfies $\partial_i\partial_{i-1}=0$ for all $0\le i\le L-1$.

If we consider vector space over $\mathbb{F}_2$, namely, $C_i:=F_2^{n_i}$, a chain complex $\mathfrak{C}$ with length 2 naturally corresponds to a CSS code $C\left(\mathfrak{C}\right)$, namely,
\begin{equation}
	C\left(\mathfrak{C}\right)=\mathbb{F}_2^{m_z}\stackrel{H_z^T}{\longrightarrow}\mathbb{F}_2^N\stackrel{H_x}{\longrightarrow}\mathbb{F}_2^{m_x}
\end{equation}
where the commutation condition $H_xH_z^T=\mathbf{0}$ is naturally satisfied.

It can be seen that a classical linear code $C=\ker{H}$ corresponds to a chain complex of length 1, namely,
\begin{equation}
	\mathbb{F}_2^N\stackrel{H}{\longrightarrow}\mathbb{F}_2^m
\end{equation}

\subsection {Hypergraph product and 3D XYZ product code construction}
\label{HP and 3D XYZ}
Hypergraph product is making use of two classical linear codes to construct a CSS code \cite{tillich2013quantum}. More precisely, given two classical linear code $C_1=\ker{H_1}$ and $C_2=\ker{H_2}$ (where $H_i, \ i\ \in\ \left\{1,\ 2\right\}$, are the parity check matrices of size $m_i\times n_i$ of codes $C_i$), hypergraph product is to construct a chain complex $\mathfrak{C}$ with length 2 as follow,
\begin{equation}
	\label{length2Chain}
	\mathfrak{C}=F_2^{m_1\times n_2}\stackrel{H_z^T}{\longrightarrow}F_2^{n_1\times n_2\oplus m_1\times m_2}\stackrel{H_x}{\longrightarrow}F_2^{n_1m_2}
\end{equation}
where $H_x = \left(I_{n_1}\otimes H_{2},\ H_{1}^{T}\otimes I_{m_2}\right)$ and $H_z = \left(H_1\otimes I_{n_2},\ I_{m_2}\otimes H_2^T\right)$. It is easy to verify that $H_xH_z^T=2H_1^T\otimes H_2=\mathbf{0}$. Thus Eq. (\ref{length2Chain}) naturally corresponds to a CSS code $C\left(\mathfrak{C}\right)$. Fig. \ref{HPchaincomplex} shows the tensor-product structure corresponding to Eq. (\ref{length2Chain}).
\begin{figure}[htbp]
	\centering
	\includegraphics[width=0.48\textwidth]{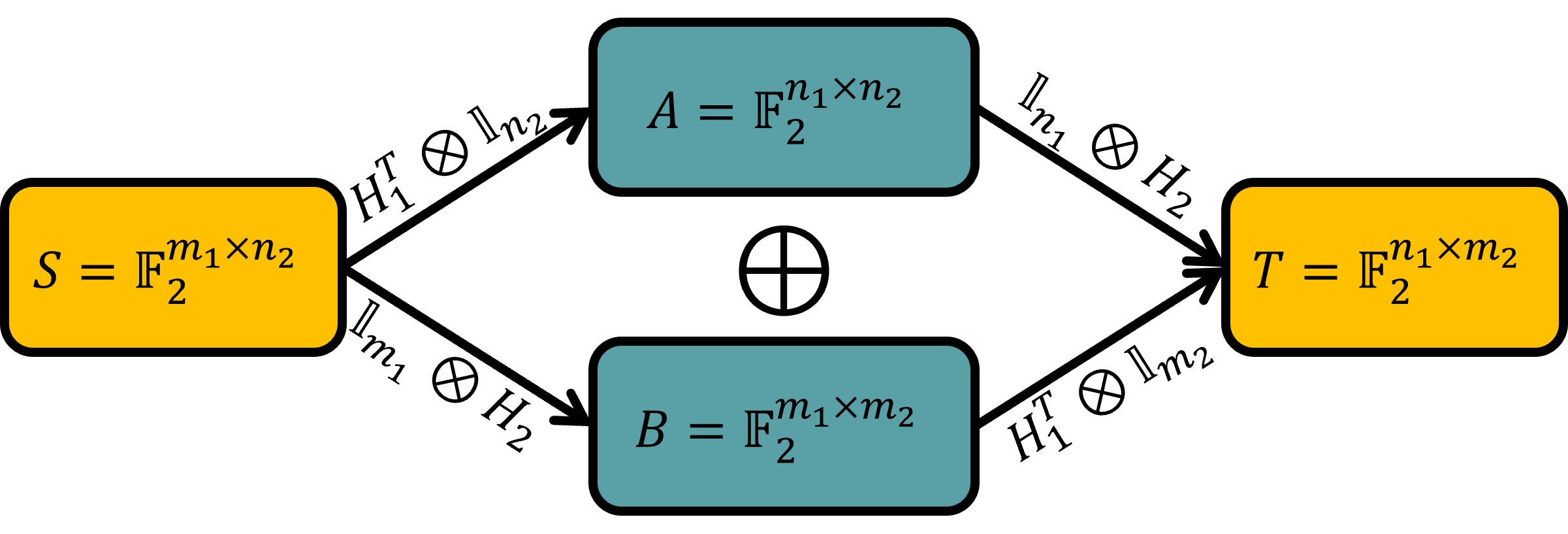}
	\caption{The tensor-product structure of the  hypergraph product of two classical linear codes $C_1=\ker{H_1}$ and $C_2=\ker{H_2}$.}
	\label{HPchaincomplex}
\end{figure}

The tensor-product structure depicted in Fig. \ref{HPchaincomplex} can derive not only a CSS code but also a non-CSS code. As shown in Fig. \ref{HPchaincomplex}, the stabilizer generators can be divided into two parts, $S$ and $T$, and each part consists of two classes of qubits, $A$ and $B$. In the case of constructing a CSS code, we assign $X$-type stabilizers to one of $S$ or $T$ and $Z$-type stabilizers to the other. If we let stabilizer generators in $S$ and $T$ to be mixed $X$ and $Z$ types, we can obtain a non-CSS code. More precisely, let the nontrivial operators of stabilizers in $S$ acting on qubits in $A$ and $B$ be Pauli $X$ and $Z$ operators, respectively. Similarly, let the nontrivial operators of stabilizers in $T$ acting on qubits in $A$ and $B$ be Pauli $Z$ and $X$ operators, respectively. In this way, we obtain the following stabilizer generator matrix,
\begin{equation}
	\mathcal{S} = 
	\begin{bmatrix}
		S\\
		T
	\end{bmatrix} = \begin{bmatrix}
		X^{\left(H_1\otimes I_{n_2}\right)},Z^{\left(I_{m_1}\otimes H_2^T\right)}\\
		Z^{\left(I_{n_1}\otimes H_2\right)},X^{\left(H_1^T\otimes I_{m_2}\right)}
	\end{bmatrix}
\end{equation}
where the notation $\mathcal{P}=P^H\ \left(P\in\left\{X,\ Y,Z\right\}\right)$ denotes a Pauli tensor, which means for any entry of matrix $H$, if the entry is 1, the tensor $\mathcal{P}$ places a Pauli operator $P$ at the corresponding position, and an identity operator $I$ otherwise. It can be verified any pair of stabilizer generators in $\mathcal{S}$ commutes.
\begin{figure}[htbp]
	\centering
	\includegraphics[width=0.48\textwidth]{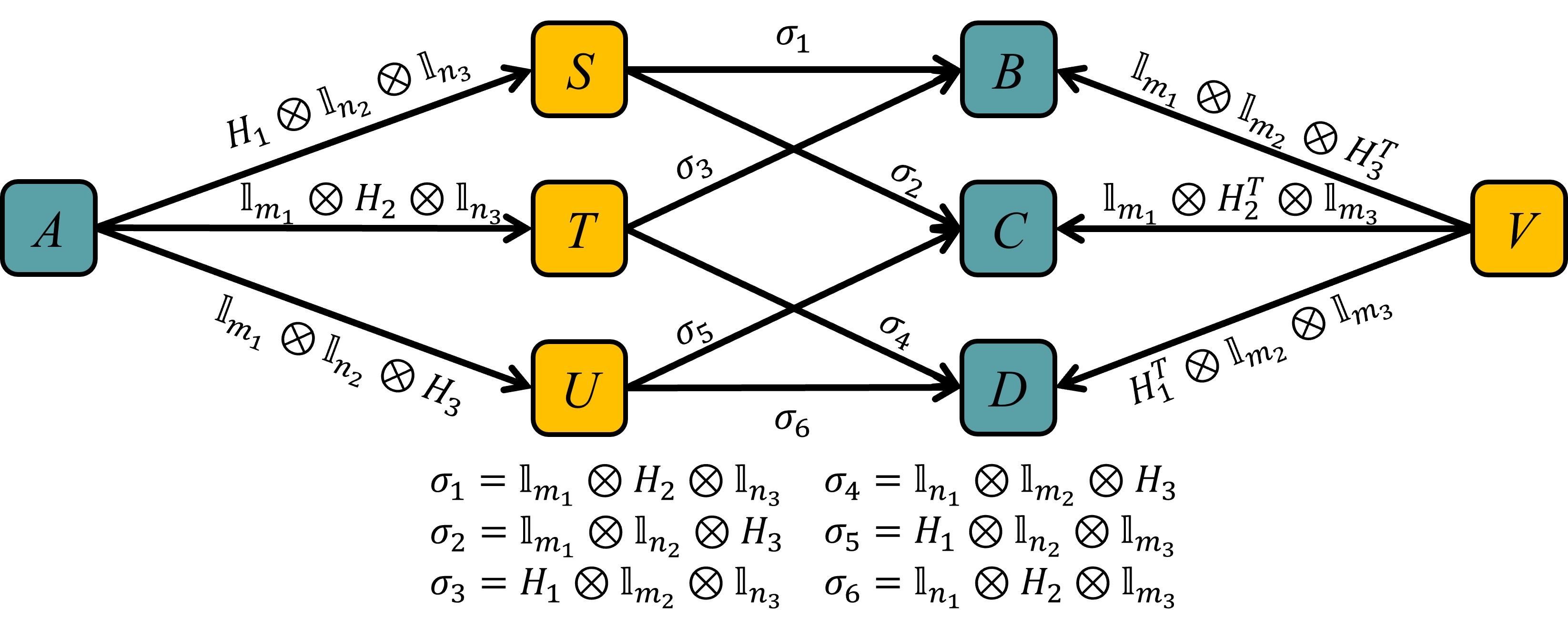}
	\caption{The tensor-product structure of the 3D XYZ product of three classical linear codes $C_1=\ker{H_1}$, $C_2=\ker{H_2}$ and $C_3=\ker{H_3}$.}
	\label{XYZchaincomplex}
\end{figure}

The 3D XYZ product \cite{leverrier2022quantum} is making use of three classical linear codes to construct a non-CSS code, which can be seen as a three-fold variant of the hypergraph product. Formally, giving three parity check matrices $H_i$ of size $m_i\times n_i\ \left(i\ =\ 1,\ 2,\ 3\right)$, one can construct the tensor-product structure as shows in Fig. \ref{XYZchaincomplex}.

In Fig. \ref{XYZchaincomplex}, $A$, $B$, $C$ and $D$ represent vector spaces which index the qubits and $S$, $T$, $U$ and $V$ represent vector spaces which index the stabilizer generators, namely,
\begin{equation}
	\begin{aligned}
		&A\in\mathbb{F}_{2}^{n_1\times n_2\times n_3},\ B\in\mathbb{F}_{2}^{m_1\times m_2\times n_3},\\
		&C\in\mathbb{F}_{2}^{m_1\times n_2\times m_3},\ D\in\mathbb{F}_{2}^{n_1\times m_2\times m_3}
	\end{aligned}	
\end{equation}
and
\begin{equation}
	\begin{aligned}
		&S\in\mathbb{F}_{2}^{m_1\times n_2\times n_3},\ T\in\mathbb{F}_{2}^{n_1\times m_2\times n_3},\\
		&U\in\mathbb{F}_{2}^{n_1\times n_2\times m_3},\ V\in\mathbb{F}_{2}^{m_1\times m_2\times m_3}
	\end{aligned}
\end{equation}

Similar with the case of constructing a non-CSS code from the tensor-product structure shown in Fig. \ref{HPchaincomplex}, the stabilizer generator matrix $\mathcal{S}$ of the corresponding 3D XYZ product code is
\begin{widetext}
	\begin{equation}
		\mathcal{S} = \begin{bmatrix}
			X^{\left(H_1\otimes I_{n_2}\otimes I_{n_3}\right)} &Y^{\left(I_{m_1}\otimes H_2^T\otimes I_{n_3}\right)} &Z^{\left(I_{m_1}\otimes I_{n_2}\otimes H_3^T\right)} &I^{\left(m_1n_2n_3\times n_1m_2m_3\right)}\\
			Y^{\left(I_{n_1}\otimes H_2\otimes I_{n_3}\right)} &X^{\left(H_1^T\otimes I_{m_2}\otimes I_{n_3}\right)} &I^{\left(n_1m_2n_3\times m_1n_2m_3\right)} &Z^{\left(I_{n_1}\otimes I_{m_2}\otimes H_3^T\right)}\\
			Z^{\left(I_{n_1}\otimes I_{n_2}\otimes H_3\right)} &I^{\left(n_1n_2m_3\times m_1m_2n_3\right)} &X^{\left(H_1^T\otimes I_{n_2}\otimes I_{m_3}\right)} &Y^{\left(I_{n_1}\otimes H_2^T\otimes I_{m_3}\right)}\\
			I^{\left(m_1m_2m_3\times n_1n_2n_3\right)} &Z^{\left(I_{m_1}\otimes I_{m_2}\otimes H_3\right)} &Y^{\left(I_{m_1}\otimes H_2\otimes I_{m_3}\right)} &X^{\left(H_1\otimes I_{m_2}\otimes I_{m_3}\right)}
		\end{bmatrix}
	\end{equation}
\end{widetext}
The code length is $N=n_1n_2n_3+m_1m_2n_3+m_1n_2m_3+n_1m_2m_3$.

So far, there has been limited study on the error-correcting performance of 3D XYZ product codes. In Sect. \ref{4}, we study the error-correcting performance of the 3D Chamon code, which is an instance of the 3D XYZ product of three repetition codes.

\begin{figure}[htbp]
	\centering
	\includegraphics[width=0.48\textwidth]{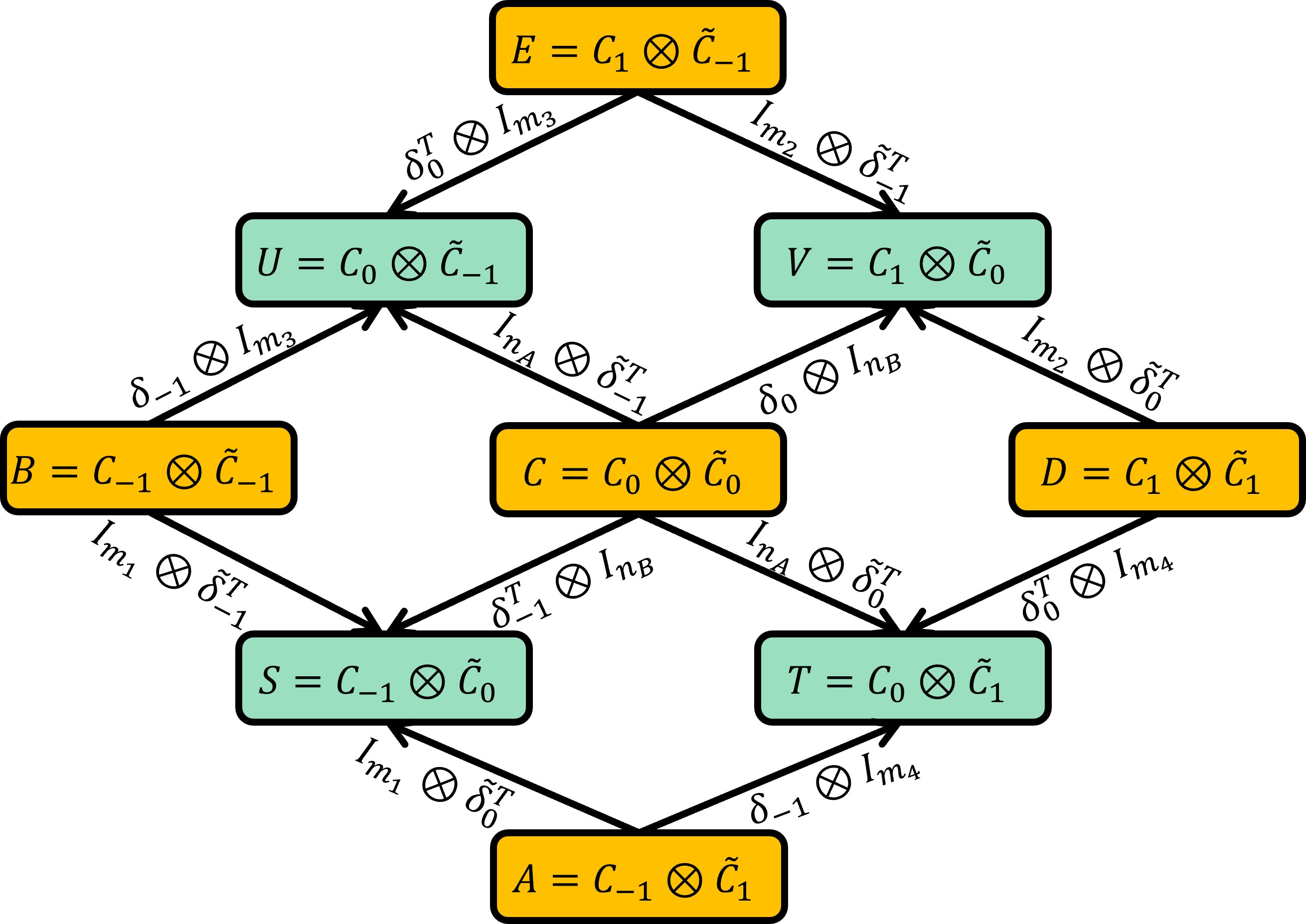}
	\caption{The tensor-product structure of the 4D XYZ product of two length-2 chain complexes $\mathfrak{C}_1=C_{-1}\stackrel{\delta_{-1}}{\longrightarrow}C_0\stackrel{\delta_{0}}{\longrightarrow}C_1$ and $\mathfrak{C}_2={\widetilde{C}}_{-1}\stackrel{\widetilde{\delta}_{-1}}{\longrightarrow}{\widetilde{C}}_0\stackrel{\widetilde{\delta}_{0}}{\longrightarrow}{\widetilde{C}}_1$.}
	\label{4DXYZproduct}
\end{figure}

\section {Four-dimensional XYZ product codes}
\label{3}
In this section, we first introduce the 4D XYZ product code construction. Second, we explain how to compute the code dimension of 4D XYZ product codes and compute two instances' code dimension. Third, we present the general form of the logical operators of 4D XYZ product codes, based on which we prove the upper bound of the code distance.

\subsection {The 4D XYZ product code construction}
\label{4D XYZ product code construction}
The 4D XYZ product construct a non-CSS code by using either two CSS codes or four classical linear codes, since one can use four classical codes to construct two CSS codes through hypergraph product or concatenation. Formally, the 4D XYZ product code construction is defined as follows.

\begin{definition}[\textbf{The 4D XYZ product code construction}]
	\label{definition4DXYZ}
	Giving two length-2 chain complexes $\mathfrak{C}_1=C_{-1}\stackrel{\delta_{-1}}{\longrightarrow}C_0\stackrel{\delta_{0}}{\longrightarrow}C_1$ and $\mathfrak{C}_2={\widetilde{C}}_{-1}\stackrel{\widetilde{\delta}_{-1}}{\longrightarrow}{\widetilde{C}}_0\stackrel{\widetilde{\delta}_{0}}{\longrightarrow}{\widetilde{C}}_1$, which corresponds to two CSS codes $C\left(\mathfrak{C}_1\right)$ and $C\left(\mathfrak{C}_2\right)$, one can construct the tensor-product structure as shown in Fig. \ref{4DXYZproduct},
	where the qubits are divided into A, B, C, D, E five parts and stabilizer generators are divided into S, T, U, V four parts. The corresponding stabilizer matrix $\mathcal{S}$ is
	\begin{widetext}
		\begin{equation}
			\label{4DXYZ stabilizer1}
			\mathcal{S} =\begin{bmatrix}
				S\\
				T\\
				U\\
				V
			\end{bmatrix} = \begin{bmatrix}
				X^{\left(I_{m_1}\otimes{\widetilde{\delta}}_0^T\right)} &Y^{\left(I_{m_1}\otimes{\widetilde{\delta}}_{-1}\right)} &Z^{\left(\delta_{-1}^T\otimes I_{n_B}\right)} &I &I\\
				Y^{\left(\delta_{-1}\otimes I_{m_4}\right)} &I &X^{\left(I_{n_A}\otimes{\widetilde{\delta}}_0\right)} &Z^{\left(\delta_0^T\otimes I_{m_4}\right)} &I\\
				I &Z^{\left(\delta_{-1}\otimes I_{m_3}\right)} &X^{\left(I_{n_A}\otimes{\widetilde{\delta}}_{-1}^T\right)} &I &Y^{\left(\delta_0^T\otimes I_{m_3}\right)}\\
				I &I &Z^{\left(\delta_0\otimes I_{n_B}\right)} &Y^{\left(I_{m_2}\otimes{\widetilde{\delta}}_{0}^{T}\right)} &X^{\left(I_{m_2}\otimes{\widetilde{\delta}}_{-1}\right)}
			\end{bmatrix}
		\end{equation}
	\end{widetext}
	where $m_1$, $m_2$, $m_3$, $m_4$, $n_A$, $n_B$ are the dimensions of vector spaces $C_{-1}$, $C_1$, ${\widetilde{C}}_{-1}$, ${\widetilde{C}}_1$, $C_0$ and ${\widetilde{C}}_0$, respectively. The total number of qubits is
	\begin{equation}
		\label{4DXYZ code length}
		N=m_1m_4+m_1m_3+n_An_B+m_2m_4+m_2m_3
	\end{equation}
\end{definition}

One can see that the tensor-product structure of 4D XYZ product in Fig. \ref{4DXYZproduct} is the same as that of 4D homological product (see Fig. 4 in Ref. \cite{Campbell_2019}). Therefore, it can be considered as a variant of the standard 4D homological product of two CSS codes.

Careful readers may wonder whether the codes constructed via the 4D XYZ product can be transformed into CSS codes by finite-depth unitary circuits. In Appendix \ref{CSS version}, we affirmatively answer this question.

Next, we prove that the resulting stabilizer group in Eq. (\ref{4DXYZ stabilizer1}) is Abelian.
\begin{corollary}
	The stabilizer group of the 4D XYZ product code is Abelian.
\end{corollary}
\begin{proof}
	To prove that the stabilizer group of the the 4D XYZ product code is Abelian, we need to prove that any pair of stabilizers commute.  Given the stabilizer matrix $\mathcal{S}$ of the 4D XYZ product code as shown in Eq. (\ref{4DXYZ stabilizer1}), the corresponding parity-check matrix $H$ in symplectic representation is $H=\left(H_x \mid H_z\right)$, where
	\begin{widetext}
		\begin{equation}
			\label{Hx}
			H_x =\begin{bmatrix}
				I_{m_1}\otimes{\widetilde{\delta}}_0^T &I_{m_1}\otimes{\widetilde{\delta}}_{-1} &\textbf{0} &\textbf{0} &\textbf{0}\\
				\delta_{-1}\otimes I_{m_4} &\textbf{0} &I_{n_A}\otimes{\widetilde{\delta}}_0 &\textbf{0} &\textbf{0}\\
				\textbf{0} &\textbf{0} &I_{n_A}\otimes{\widetilde{\delta}}_{-1}^T &\textbf{0} &\delta_0^T\otimes I_{m_3}\\
				\textbf{0} &\textbf{0} &\textbf{0} &I_{m_2}\otimes{\widetilde{\delta}}_{0}^{T} &I_{m_2}\otimes{\widetilde{\delta}}_{-1}
			\end{bmatrix}
		\end{equation}
	\end{widetext}
	and
	\begin{widetext}
		\begin{equation}
			\label{Hz}
			H_z = \begin{bmatrix}
				\textbf{0} &I_{m_1}\otimes{\widetilde{\delta}}_{-1} &\delta_{-1}^T\otimes I_{n_B} &\textbf{0} &\textbf{0}\\
				\delta_{-1}\otimes I_{m_4} &\textbf{0} &\textbf{0} &\delta_0^T\otimes I_{m_4} &\textbf{0}\\
				\textbf{0} &\delta_{-1}\otimes I_{m_3} &\textbf{0} &\textbf{0} &\delta_0^T\otimes I_{m_3}\\
				\textbf{0} &\textbf{0} &\delta_0\otimes I_{n_B} &I_{m_2}\otimes{\widetilde{\delta}}_{0}^{T} &\textbf{0}
			\end{bmatrix}
		\end{equation}
	\end{widetext}
Then we have
	\begin{widetext}
		\begin{equation}
			H_xH_z^T+H_zH_x^T =2\begin{bmatrix}
				{\widetilde{\delta}}_{-1}{\widetilde{\delta}}_{-1}^T & \delta_{-1}^T\otimes{\widetilde{\delta}}_0^T &\delta_{-1}^T\otimes{\widetilde{\delta}}_{-1} &\textbf{0}\\
				\delta_{-1}\otimes{\widetilde{\delta}}_0 &\delta_{-1}\delta_{-1}^T &\textbf{0} &\delta_0^T\otimes{\widetilde{\delta}}_0^T\\
				\delta_{-1}\otimes{\widetilde{\delta}}_{-1}^T &\textbf{0} &\delta_0^T\delta_0 &\delta_0^T\otimes{\widetilde{\delta}}_{-1}^T\\
				\textbf{0} &\delta_0\otimes{\widetilde{\delta}}_0^T &\delta_0\otimes{\widetilde{\delta}}_{-1} &{\widetilde{\delta}}_0^T{\widetilde{\delta}}_0
			\end{bmatrix}=\textbf{0}
		\end{equation}
	\end{widetext}
	which means any pair of stabilizers commute and the proof is completed.
\end{proof}

Next, we construct a 4D XYZ product code from two quantum concatenated codes. Each concatenated code is derived from a pair of repetition codes with respective block lengths $\left(n_1,\ n_2\right)$ and $\left(n_3,\ n_4\right)$, where $n_1$, $n_2$, $n_3$ and $n_4$ are all odd integers.

\subsection {The code dimension of 4D XYZ product codes}
\label{Dimension}
In this section, we compute the code dimension of 4D XYZ product codes. Similar to Ref. \cite{leverrier2022quantum}, this problem reduces to finding the number of independent stabilizer generators. First, we give a general solution to this problem in \textbf{Theorem} \ref{The dimension of 4D XYZ product code}. Second, we consider two instances of the 4D XYZ product. The first one is constructed from two concatenated codes. Each concatenated code is derived from a pair of repetition codes with respective block lengths $\left(n_1,\ n_2\right)$ and $\left(n_3,\ n_4\right)$, where $n_1$, $n_2$, $n_3$ and $n_4$ are all odd integers. We prove its code dimension is $1$ in textbf{Corollary} \ref{dimension of 4D XYZ concatenated}. The second one is constructed from two hypergraph product codes, which are also obtained from two pairs of repetition codes with block lengths $\left(n_1,\ n_2\right)$ and $\left(n_3,\ n_4\right)$, and its code dimension is proven to be $8\gcd(n_1, n_2)\gcd(n_3, n_4)$ in \textbf{Corollary} \ref{dimension of the 4D Chamon code}. We refer to these two codes as the 4D XYZ concatenated code and the 4D Chamon code (the rationale for this name is discussed in Sect. \ref{4D Chamon code}), respectively.

For convenience, we rewrite Eq. (\ref{4DXYZ stabilizer1}) as
\begin{widetext}
	\begin{equation}
		\label{4DXYZ stabilizer2}
		\mathcal{S} =\begin{bmatrix}
			S\\
			T\\
			U\\
			V
		\end{bmatrix} = \begin{bmatrix}
			X^{\left(I_{m_1}\otimes H_{x_2}^T\right)} &Y^{\left(I_{m_1}\otimes H_{z_2}^T\right)} &Z^{\left(H_{z_1}\otimes I_{n_B}\right)} &I &I\\
			Y^{\left(H_{z_1}^T\otimes I_{m_4}\right)} &I &X^{\left(I_{n_A}\otimes H_{x_2}\right)} &Z^{\left(H_{x_1}^T\otimes I_{m_4}\right)} &I\\
			I &Z^{\left(H_{z_1}^T\otimes I_{m_3}\right)} &X^{\left(I_{n_A}\otimes H_{z_2}\right)} &I &Y^{\left(H_{x_1}^T\otimes I_{m_3}\right)}\\
			I &I &Z^{\left(H_{x1}\otimes I_{n_B}\right)} &Y^{\left(I_{m_2}\otimes H_{x_2}^T\right)} &X^{\left(I_{m_2}\otimes H_{z_2}^T\right)}
		\end{bmatrix}
	\end{equation}
\end{widetext}
where $H_{x_1}=\delta_0$ ($H_{x_2}={\widetilde{\delta}}_0$) and $H_{z_1}=\delta_{-1}^T$ ($H_{z_2}={\widetilde{\delta}}_{-1}^T$) are the $X$-type and $Z$-type parity-check matrices of code $C\left(\mathfrak{C}_1\right)$ $\left(C\left(\mathfrak{C}_2\right)\right)$, respectively. Recall that the dimensions of $H_{x_1}$, $H_{z_1}$, $H_{x_2}$ and $H_{z_2}$ are $m_1\times n_A$, $m_2\times n_A$, $m_3\times n_B$ and $m_4\times n_B$, respectively. 

For the remainder of this paper, bold lowercase letters represent column vectors.
\begin{theorem}[\textbf{The code dimension of 4D XYZ product code}]
	\label{The dimension of 4D XYZ product code}
	The code dimension $k$ of 4D XYZ product code is $\left(n_A-m_1-m_2\right)\left(n_B-m_3-m_4\right)+k_{SV}+k_{TU}$, where $k_{SV}=\dim\left(\ker\left(\left[H_{z_1}^T,H_{x_1}^T\right]\right)\otimes \ker\left(\begin{bmatrix}H_{x_2}\\H_{z_2}\end{bmatrix}\right)\right)$ is the number of independent solutions $\begin{bmatrix}\textbf{s}\\ \textbf{v}\end{bmatrix}$ of
	\begin{equation}
		\label{sv}
		\begin{aligned}
			&I_{m_1}\otimes H_{x_2}\textbf{s}=I_{m_1}\otimes H_{z_2}\textbf{s}=\textbf{0}\\
			&I_{m_2}\otimes H_{x_2}\textbf{v}=I_{m_2}\otimes H_{z_2}\textbf{v}=\textbf{0}\\                      
			&H_{z_1}^T\otimes I_{n_B}\textbf{s}+H_{x_1}^T\otimes I_{n_B}\textbf{v}=\textbf{0}
		\end{aligned}		
	\end{equation}
	and $k_{TU}=\dim\left(\ker\left(\begin{bmatrix}H_{x_1}\\H_{z_1}\end{bmatrix}\right)\otimes \ker\left(\left[H_{z_2}^T,H_{x_2}^T\right]\right) \right)$ is the number of independent solutions $\begin{bmatrix}\textbf{t}\\ \textbf{u}\end{bmatrix}$ of
	\begin{equation}
		\label{tu}
		\begin{aligned}
			&H_{z_1}\otimes I_{m_4}\textbf{t}=H_{x_1}\otimes I_{m_4}\textbf{t}=\textbf{0}\\
			&H_{z_1}\otimes I_{m_3}\textbf{u}=H_{z_1}\otimes I_{m_3}\textbf{u}=\textbf{0}\\                      
			&I_{n_A}\otimes H_{x_2}^T \textbf{u}=I_{n_A}\otimes H_{z_2}^T\textbf{u}=\textbf{0}
		\end{aligned}		
	\end{equation}
\end{theorem}
\begin{proof}
	The total number $m$ of stabilizer generators in Eq. (\ref{4DXYZ stabilizer2}) is $m=m_1n_B+m_2n_B+m_3n_A+m_4n_A$. Thus, according to Eq. (\ref{4DXYZ code length}), the 4D XYZ product code encodes at least $N-m=\left(n_A-m_1-m_2\right)\left(n_B-m_3-m_4\right)$ logical qubits. However, we should notice that stabilizer generators in Eq. (\ref{4DXYZ stabilizer2}) may not be independent of each other. Thus, we should find out the number of independent stabilizer generators in Eq. (\ref{4DXYZ stabilizer2}).
	
	Observing Eq. (\ref{4DXYZ stabilizer2}), one can see that the stabilizer generators are divided into $S$, $T$, $U$, $V$ four parts, and any stabilizer in $S$ or $V$ cannot be generated by the product of some stabilizers in $T$ and $U$. Thus, we first consider the independence of the stabilizers in $S$ and $V$, with a similar analysis applied to the stabilizers in T and U.
	
	Considering stabilizers in $S$ and $V$ simultaneously, one can see that the product of some stabilizers in $S$ might be equal to the product of some stabilizers in $V$, namely,
	\begin{equation}
		\begin{aligned}
			&\left[I_{m_1}\otimes H_{x_2}^T,\ I_{m_1}\otimes H_{z_2}^T,\ H_{z_1}\otimes I_{n_B},\ \textbf{0},\ \textbf{0}\right]^T\textbf{\emph{s}}\\
			&=\left[\textbf{0},\ \textbf{0},\ H_{x_1}\otimes I_{n_B},\ I_{m_2}\otimes H_{x_2}^T,\ I_{m_2}\otimes H_{z_2}^T\right]^T\textbf{\emph{v}}
		\end{aligned}
	\end{equation}
	Then we have Eq. (\ref{sv}).
	
	Our goal is to compute $k_{SV}$ and $k_{TU}$. Observing that equation systems (\ref{sv}) and (\ref{tu}) are similar, thus we only need to compute $k_{SV}$, and $k_{TU}$ will follow a similar fashion.
	
	It can be seen that the solution of $I_{m_1}\otimes H_{x_2}\textbf{\emph{s}}=I_{m_1}\otimes H_{z_2}\textbf{\emph{s}}=\mathbf{0}$ is $\textbf{\emph{s}}\in\mathcal{C}_{m_1}\otimes \ker\left(\begin{bmatrix}H_{x_2}\\H_{z_2}\end{bmatrix}\right)$ and the solution of $I_{m_2}\otimes H_{x_2}\textbf{\emph{v}}=I_{m_2}\otimes H_{z_2}\textbf{\emph{v}}=\mathbf{0}$ is $\textbf{\emph{v}}\in\mathcal{C}_{m_2}\otimes \ker\left(\begin{bmatrix}H_{x_2}\\H_{z_2}\end{bmatrix}\right)$. Similarly, the solution of $H_{z_1}^T\otimes I_{n_B}\textbf{\emph{s}}=H_{x_1}^T\otimes I_{n_B}\textbf{\emph{v}}$ is $\begin{bmatrix}\textbf{\emph{s}}\\ \textbf{\emph{v}}\end{bmatrix}\in \ker\left(\left[H_{z_1}^T,H_{x_1}^T\right]\right)\otimes \mathcal{C}_{n_B}$. Thus, the solution of equation system (\ref{sv}) is $\begin{bmatrix}\textbf{\emph{s}}\\ \textbf{\emph{v}}\end{bmatrix}\in \ker\left(\left[H_{z_1}^T,H_{x_1}^T\right]\right)\otimes \ker\left(\begin{bmatrix}H_{x_2}\\H_{z_2}\end{bmatrix}\right)$, and the dimension of $\ker\left(\left[H_{z_1}^T,H_{x_1}^T\right]\right)\otimes \ker\left(\begin{bmatrix}H_{x_2}\\H_{z_2}\end{bmatrix}\right)$ is the number of independent solutions $\begin{bmatrix}\textbf{\emph{s}}\\ \textbf{\emph{v}}\end{bmatrix}$ of equation system (\ref{sv}).
	
	Notice that
	\begin{equation}
		k_{SV}=\dim\left(\ker\left(\left[H_{z_1}^T,H_{x_1}^T\right]\right)\otimes \ker\left(\begin{bmatrix}H_{x_2}\\H_{z_2}\end{bmatrix}\right)\right)
	\end{equation}
	and according to the theory of linear algebra, the number of independent stabilizers in $S$ and $V$, $r_{SV}$, is
	\begin{equation}
		r_{SV}=m_1n_B+m_2n_B-k_{SV}
	\end{equation}
	
	Similarly, notice that $k_{TU}=\dim\left(\ker\left(\begin{bmatrix}H_{x_1}\\H_{z_1}\end{bmatrix}\right)\otimes \ker\left(\left[H_{z_2}^T,H_{x_2}^T\right]\right) \right)$ is the number of independent solutions $\begin{bmatrix}\textbf{\emph{t}}\\ \textbf{\emph{u}}\end{bmatrix}$ of equation system (\ref{tu}), thus the number of independent stabilizers in $T$ and $U$, $r_{TU}$, is
	\begin{equation}
		r_{TU}=m_3n_A+m_4n_A-k_{TU}
	\end{equation}
	
	Thus, the total number of independent stabilizers in Eq. (\ref{4DXYZ stabilizer2}) is $r_{SV}+r_{TU}$, and the code dimension of 4D XYZ product code is
	\begin{equation}
		\label{dimension equation}
		\begin{aligned}
			&N-r_{SV}-r_{TU}\\
			&=\left(n_A-m_1-m_2\right)\left(n_B-m_3-m_4\right)+k_{SV}+k_{TU}
		\end{aligned}
	\end{equation}
\end{proof}

Next, we construct a 4D XYZ product code from two quantum concatenated codes. Each concatenated code is derived from a pair of repetition codes with respective block lengths $\left(n_1,\ n_2\right)$ and $\left(n_3,\ n_4\right)$, where $n_1$, $n_2$, $n_3$ and $n_4$ are all odd integers. To characterize this structure, let $C_1$ denote the length-$n_1$ repetition code (outer code) with parity-check matrix  $H_1$ of size $(n_1-1)\times n_1$, and $C_2$ the length-$n_2$ repetition code (inner code) with parity-check matrix $H_2$ of size $(n_2-1)\times n_2$. If we utilize $C_1$ as the phase-flip code which can only correct Pauli $Z$ errors and $C_2$ as the bit-flip code which can only correct Pauli $X$ errors, the corresponding $X$-type and $Z$-type parity-check matrices $H_x$ and $H_z$ of the concatenated code obtained from $C_1$ and $C_2$ are then given by:
\begin{equation}
	\label{concatenated code Hx}
	H_x = \begin{bmatrix}
		\overbrace{1,\cdots,1}^{n_2}, &\overbrace{1,\cdots,1}^{n_2},  &\overbrace{0,\cdots,0}^{(n_1 - 2)n_2}\\
		\vdots&\ddots&\vdots\\
		\underbrace{0,\cdots,0}_{(n_1 - 2)n_2}, &\underbrace{1,\cdots,1}_{n_2},  &\underbrace{1,\cdots,1}_{n_2}
	\end{bmatrix}
\end{equation}
\begin{equation}
	\label{concatenated code Hz}
	H_z=\begin{bmatrix}H_1&\cdots&\textbf{0}\\\vdots&\ddots&\vdots\\ \textbf{0}&\cdots&H_1\\\end{bmatrix}
\end{equation}

In Appendix \ref{Code concatenation}, we provide more details about code concatenation.
\begin{corollary}[\textbf{The code dimension of the 4D XYZ product concatenated code}]
	\label{dimension of 4D XYZ concatenated}
	The code dimension $k$ of the 4D XYZ product of two concatenated codes, each of which is derived from a pair of repetition codes with respective block lengths $\left(n_1,\ n_2\right)$ and $\left(n_3,\ n_4\right)$, is $1$. Here $n_1$, $n_2$, $n_3$ and $n_4$ are all odd integers.
\end{corollary}
\begin{proof}
	According to Theorem \ref{The dimension of 4D XYZ product code} and $\left(n_A-m_1-m_2\right)\left(n_B-m_3-m_4\right)=1$, the code dimension $k$ of the 4D XYZ product concatenated code is $1+k_{SV}+k_{TU}$, where $k_{SV}=\dim\left(\ker\left(\left[H_{z_1}^T,H_{x_1}^T\right]\right)\otimes \ker\left(\begin{bmatrix}H_{x_2}\\H_{z_2}\end{bmatrix}\right)\right)$ and $k_{TU}=\dim\left(\ker\left(\begin{bmatrix}H_{x_1}\\H_{z_1}\end{bmatrix}\right)\otimes \ker\left(\left[H_{z_2}^T,H_{x_2}^T\right]\right) \right)$.
	
	Observing that $\dim\left(\ker\left(\begin{bmatrix}H_{x}\\H_{z}\end{bmatrix}\right)\right)$ can be seen as the number of independent solutions $\hat{\textbf{\emph{e}}}=\left(\textbf{\emph{e}},\textbf{\emph{e}}\right)^T$ of 
	\begin{equation}
		\begin{pmatrix}
			H_x &\textbf{0}\\
			\textbf{0} &H_z
		\end{pmatrix}\left(\textbf{\emph{e}},\textbf{\emph{e}}\right)^T=\begin{pmatrix}
			\textbf{0}\\
			\textbf{0}
		\end{pmatrix}
	\end{equation}
	Here, the solution $\hat{\textbf{\emph{e}}}$ can be interpreted as the symplectic representation of undetectable $Y$-type errors, which include $Y$-type stabilizers and logical operators of the corresponding code.
	
	For a concatenated code $C$ obtained from a pair of repetition codes, whose $X$-type and $Z$-type parity-check matrices are $H_x$ and $H_z$ as shown in Eq. (\ref{concatenated code Hx}) and Eq. (\ref{concatenated code Hz}) respectively, it is easy to prove that $\dim\left(\ker\left(\begin{bmatrix}H_{x}\\H_{z}\end{bmatrix}\right)\right)=1$ and $\dim\left(\ker\left(\left[H_{z}^T,H_{x}^T\right]\right)\right)=n-1-\dim\left(row\left(\begin{bmatrix}H_{x}\\H_{z}\end{bmatrix}\right)\right)=n-1-(n-1)=0$. The reason is that the undetectable $Y$-type error of code $C$ only contains one $Y$-type logical operator whose weight is code length. Thus, the code dimension $k$ of 4D XYZ product concatenated code is 1.
\end{proof}

Next, we consider the second instance of the 4D XYZ product, the 4D Chamon code, and compute its code dimension. Before computing it, we first prove \textbf{Lemma} \ref{2DToricCode} which will be used later.

\begin{lemma}[\textbf{The number of independent undetectable $Y$-type errors of the two-dimensional (2D) toric code}]
	\label{2DToricCode}
	The number of the independent $Y$-type stabilizers and logical operators of the 2D toric code is $2\gcd\left(j,k\right)$, where $j$ and $k$ are dimensions of the toric code lattice.
\end{lemma}
\begin{proof}
	Without loss of generality, we assume that $j\le k$, and the proof consists of two parts:
	
	1. Proving that the number of independent $Y$-type logical operators of 2D toric code is 2.
	
	2. Proving that the number of independent $Y$-type stabilizers of 2D toric code is $2[\gcd\left(j,k\right)-1]$.
	
	Since 2D toric code encodes 2 logical qubits, there are 2 independent $Y$-type logical operators at most. When the toric code lattice is a square, these two independent $Y$-type logical operators are full diagonal of Pauli $Y$ operators from the upper right to the lower left and from the upper left to the lower right, respectively. In Fig. \ref{logicalY2DToricCode}, we give an example of $Y$-type logical operators of 2D toric code with lattice size $3\times3$.
	\begin{figure}[htbp]
		\centering
		\includegraphics[width=0.4\textwidth]{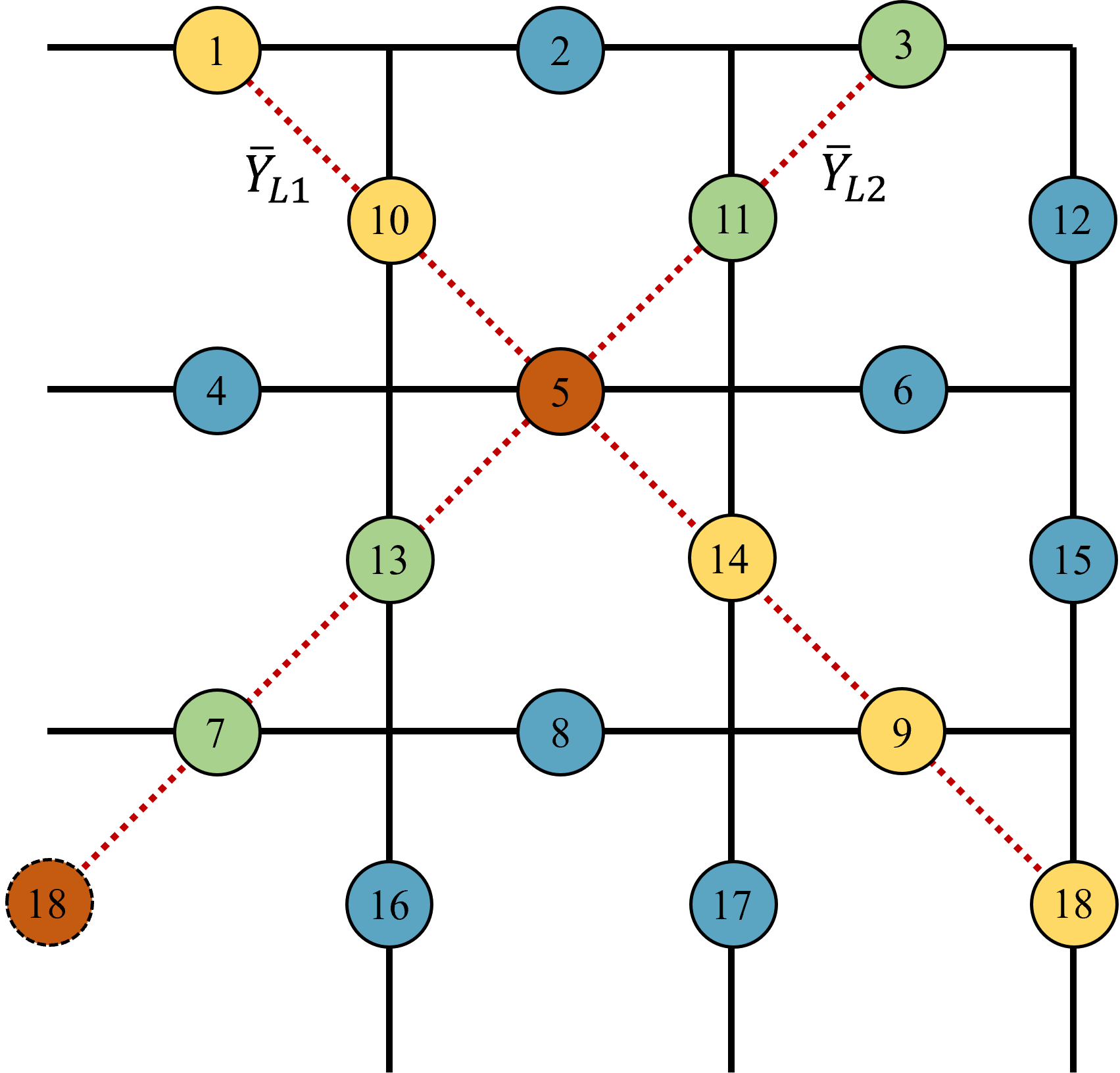}
		\caption{$Y$-type logical operators of 2D the toric code with lattice size $3\times3$.}
		\label{logicalY2DToricCode}
	\end{figure}
	
	Any $Y$-type stabilizer must be generated by the product of some $X$-type and $Z$-type stabilizers both of which act on the same set of qubits. According to the topology of toric code, this can only be realized by selecting plaquettes ($Z$-type stabilizers) and vertices ($X$-type stabilizers) along the diagonal of $j\times j$ square lattices periodically. Fig. \ref{stabilizerY2DToricCode} gives an example of $Y$-type stabilizer of toric code with size of $j=4,\ k=6$.
	\begin{figure}[htbp]
		\centering
		\includegraphics[width=0.48\textwidth]{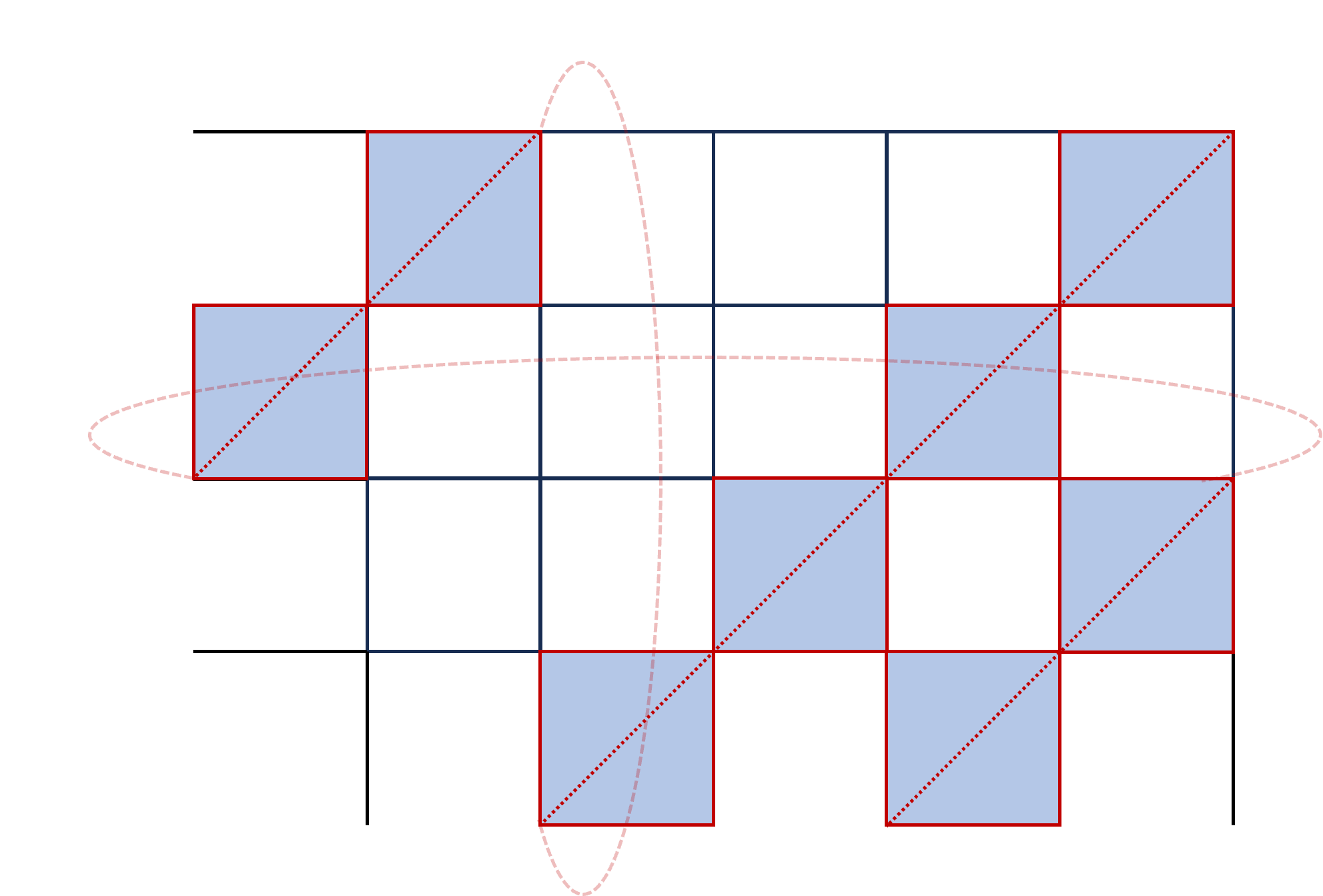}
		\caption{A $Y$-type stabilizer of toric code with lattice size $j=4,\ k=6$. The red solid lines represent the qubits which the $Y$-type stabilizer acts on. The blue-grey plaquettes denote the selected $Z$-type stabilizers, while the top-right and bottom-left vertices indicate the selected $X$-type stabilizers.}
		\label{stabilizerY2DToricCode}
	\end{figure}
	
	When $j$ and $k$ are coprime, namely, $\gcd\left(j,k\right)=1$, it is impossible to select some $X$-type and $Z$-type stabilizers to generate a $Y$-type stabilizer, thus the number of $Y$-type stabilizers is zero. While in the case of $\gcd\left(j,k\right)=g\neq1$, using the selection method as shown in Fig. \ref{stabilizerY2DToricCode}, one can find $2\left(g-1\right)$ independent $Y$-type stabilizers by shifting to left or right, where the factor 2 comes from the left-right symmetry of toric code.
	
	To sum up, the total number of the independent $Y$-type stabilizers and logical operators of the 2D toric code is $2\gcd\left(j,k\right)$, and the proof is completed.
\end{proof}

\begin{corollary}[\textbf{The code dimension of the 4D Chamon code}]
	\label{dimension of the 4D Chamon code}
	The code dimension $k$ of the 4D Chamon code, which is the 4D XYZ product of two hypergraph product codes that are obtained from two pairs of repetition codes with block lengths $\left(n_1,\ n_2\right)$ and $\left(n_3,\ n_4\right)$, is
	\begin{equation}
		k=8\gcd(n_1, n_2)\gcd(n_3, n_4)
	\end{equation}
\end{corollary}
\begin{proof}
	The hypergraph product of a pair of repetition codes is a 2D toric code\cite{tillich2013quantum}. Thus, given two toric codes which are obtained from two pairs of repetition codes with block lengths $\left(n_1,\ n_2\right)$ and $\left(n_3,\ n_4\right)$, whose parity-check matrices are $H_1$, $H_2$, $H_3$ and $H_4$, respectively, by \textbf{Definition} \ref{definition4DXYZ}, the corresponding stabilizer matrix $\textbf{\emph{S}}$ of the 4D Chamon code is as shown in Eq. (\ref{4DXYZ stabilizer2}), where $H_{x_1}=\left[I_{n_1}\otimes H_2,H_1^T\otimes I_{n_2}\right]$ with dimension of $m_1\times n_A$ and $H_{z_1}=\left[H_1\otimes I_{n_2},\ I_{n_1}\otimes H_2^T\right]$ with dimension of $m_2\times n_A$ are the $X$-type and $Z$-type parity-check matrices of toric code $C_1$, respectively, and $H_{x_2}=\left[I_{n_3}\otimes H_4,H_3^T\otimes I_{n_4}\right]$ with dimension of $m_3\times n_B$ and $H_{z_2}=\left[H_3\otimes I_{n_4},\ I_{n_3}\otimes H_4^T\right]$ with dimension of $m_4\times n_B$ are the $X$-type and $Z$-type parity-check matrices of toric code $C_2$, respectively. Recall that $n_A=2n_1n_2=m_1+m_2$ and $n_B=2n_3n_4=m_3+m_4$ is the code length of $C_1$ and $C_2$, respectively.
	
	According to \textbf{Theorem} \ref{The dimension of 4D XYZ product code} and $\left(n_A-m_1-m_2\right)\left(n_B-m_3-m_4\right)=0$, the code dimension $k$ of the 4D Chamon code is $k_{SV}+k_{TU}$, where $k_{SV}=\dim\left(\ker\left(\left[H_{z_1}^T,H_{x_1}^T\right]\right)\otimes \ker\left(\begin{bmatrix}H_{x_2}\\H_{z_2}\end{bmatrix}\right)\right)$ and $k_{TU}=\dim\left(\ker\left(\begin{bmatrix}H_{x_1}\\H_{z_1}\end{bmatrix}\right)\otimes \ker\left(\left[H_{z_2}^T,H_{x_2}^T\right]\right) \right)$.
	
	Similar with the proof in \textbf{Corollary} \ref{dimension of 4D XYZ concatenated}, $\ker\left(\begin{bmatrix}H_{x_2}\\H_{z_2}\end{bmatrix}\right)$ is the number of undetectable $Y$-type errors, which include $Y$-type stabilizers and logical operators of the toric code $C_2$. According to \textbf{Lemma} 1, the number of independent solutions $\hat{\textbf{\emph{e}}}$ is $2\gcd(n_3, n_4)$.
	
	For $\dim{\left(\ker{\left(\left[H_{z1}^T,\ \ H_{x1}^T\right]\right)}\right)}$, it is equal to $n_A-\dim{\left(Im\left(\left[H_{z_1}^T,\ \ H_{x_1}^T\right]\right)\right)}=n_A-row\left(\begin{bmatrix}H_{x_1}\\H_{z_1}\end{bmatrix}\right)=\ker\left(\begin{bmatrix}H_{x_1}\\H_{z_1}\end{bmatrix}\right)=2\gcd\left(n_1,n_2\right)$, where $Im\left(\left[H_{z_1}^T,\ \ H_{x_1}^T\right]\right)$ is the image space of $\left[H_{z_1}^T,\ \ H_{x_1}^T\right]$ and $row\left(\begin{bmatrix}H_{x_1}\\H_{z_1}\end{bmatrix}\right)$ is the row space of $\begin{bmatrix}H_{x_1}\\H_{z_1}\end{bmatrix}$.
	
	Observing that the dimension $k_{SV}$ of $\ker\left(\left[H_{z_1}^T,H_{x_1}^T\right]\right)\otimes \ker\left(\begin{bmatrix}H_{x_2}\\H_{z_2}\end{bmatrix}\right)$ is $\dim\left(\ker\left(\left[H_{z_1}^T,H_{x_1}^T\right]\right) \right)\times \dim\left(\ker\left(\begin{bmatrix}H_{x_2}\\H_{z_2}\end{bmatrix}\right)\right)$, thus $k_{SV}=4\gcd(n_1,n_2)\gcd(n_3,n_4)$. Similarly, the dimension $k_{TU}$ of $\dim\left(\ker\left(\begin{bmatrix}H_{x_1}\\H_{z_1}\end{bmatrix}\right)\otimes \ker\left(\left[H_{z_2}^T,H_{x_2}^T\right]\right) \right)$ is also $4\gcd(n_1,n_2)\gcd(n_3,n_4)$. Thus, the code dimension $k$ of the 4D Chamon code is $k=k_{SV}+k_{TU}=8\gcd(n_1,n_2)\gcd(n_3,n_4)$, which completes the proof.
\end{proof}

For the 4D toric code which is the 4D homological product of two 2D toric codes, its code dimension remains $6$ regardless of the values of $n_1$, $n_2$, $n_3$ and $n_4$. While the code dimension of the 4D Chamon code, which is $8\gcd(n_1,n_2)\gcd(n_3,n_4)$, is higher than that of the 4D toric code, which implies that using two identical component CSS codes, the 4D XYZ product may construct codes with higher code dimensions than the 4D homological product.

\subsection {Logical operators and code distance of 4D XYZ product codes}
\label{3.3}
In this section, we study the logical operators and the code distance of 4D XYZ product codes. First, we present the general forms of the logical operators in \textbf{Theorem} \ref{The logical operators of 4D XYZ product code}, followed by the proof of the upper bound of their minimum weight in \textbf{corollary} \ref{code distance proof}.

Here we define a permutation notation $\pi\left(\cdot\right)$ as
\begin{equation}
	\pi\left(\begin{bmatrix}x_1\textbf{\emph{a}}\\ \vdots\\ x_{m_1}\textbf{\emph{a}}\\ x_1\textbf{\emph{b}}\\ \vdots\\ x_{m_1}\textbf{\emph{b}}\\ y_1\textbf{\emph{a}}\\ \vdots\\ y_{m_2}\textbf{\emph{a}}\\ y_1\textbf{\emph{b}}\\ \vdots\\ y_{m_2}\textbf{\emph{b}}\end{bmatrix}\right) = \begin{bmatrix}x_1\textbf{\emph{a}}\\ x_1\textbf{\emph{b}}\\ \vdots\\ x_{m_1}\textbf{\emph{a}}\\x_{m_1}\textbf{\emph{b}}\\ y_1\textbf{\emph{a}}\\ y_1\textbf{\emph{b}}\\ \vdots\\ y_{m_2}\textbf{\emph{a}}\\y_{m_2}\textbf{\emph{b}}\end{bmatrix}
\end{equation}
which will be used in the proof of \textbf{Theorem \ref{The logical operators of 4D XYZ product code}}.

\begin{theorem}[\textbf{The general forms of the logical operators of 4D XYZ product codes}]
	\label{The logical operators of 4D XYZ product code}
	Given two CSS codes $C_1$ and $C_2$, whose $X$-type and $Z$-type parity-check matrices are $H_{x_1}$, $H_{z_1}$ and $H_{x_2}$, $H_{z_2}$, respectively, the logical operators of the corresponding 4D XYZ product code are divided into two types. The first type of logical operators are
	\begin{equation}
		\begin{aligned}
			&X_L=\left[I\ \ \ \ \ I\ \ \ \ X\left(\textbf{r}^T\right)\ \ \ I\ \ \ \ \ I\right]\\
			&Z_L=\left[I\ \ \ \ \ I\ \ \ \ Z\left(\textbf{w}^T\right)\ \ \ I\ \ \ \ \ I\right]
		\end{aligned}		
	\end{equation}
	where
	\begin{equation}
		\begin{aligned}
			\textbf{r}\in \ker\left(\begin{bmatrix}H_{x_1}\\H_{z_1}\end{bmatrix}\right)\otimes \left(\mathcal{C}_{n_B}\backslash Im\left(\left[H_{x_2}^T,H_{z_2}^T\right]\right)\right)\\
			\textbf{w}\in\left(\mathcal{C}_{n_A}\backslash Im\left(\left[H_{x_1}^T,H_{z_1}^T\right]\right)\right) \otimes \ker\left(\begin{bmatrix}H_{x_2}\\H_{z_2}\end{bmatrix}\right)
		\end{aligned}		
	\end{equation}
	
	The second type of logical operators are
	\begin{equation}
		\label{second logical operators}
		\begin{bmatrix}X_{L_2}\\X_{L_3}\\Z_{L_2}\\Z_{L_3}\end{bmatrix}=\begin{bmatrix}X(\textbf{a}_1^T)\ &I\ &I\ &Y(\textbf{b}_1^T)\ &I\\I\ &Y(\textbf{a}_2^T)\ &I\ &I\ &X(\textbf{b}_2^T)\\Y(\textbf{c}_1^T)\ &Z(\textbf{d}_1^T)\ &I\ &I\ &I\\I\ &I\ &I\ &Z(\textbf{c}_2^T)\ &Y(\textbf{d}_2^T)\end{bmatrix}
	\end{equation}
	where
	\begin{equation}
		\begin{aligned}
			&{\pi\left(\textbf{a}_1^T,\ \textbf{a}_2^T,\textbf{b}_1^T,\textbf{b}_2^T\right)}^T\\
			&\in \ker\left(\left[H_{z_1}^T,H_{x_1}^T\right]\right)\otimes\left(\mathcal{C}_{m_3+m_4}\backslash Im\left(\begin{bmatrix}H_{x_2}\\H_{z_2}\end{bmatrix}\right)\right)\\
			&{\pi\left(\textbf{c}_1^T,\ \textbf{d}_1^T,\textbf{c}_2^T,\textbf{d}_2^T\right)}^T\\
			&\in\left(\mathcal{C}_{m_1+m_2}\backslash Im\left(\begin{bmatrix}H_{z_1}\\H_{x_1}\end{bmatrix}\right)\right)\otimes\ker\left(\left[H_{x_2}^T,H_{z_2}^T\right]\right)
		\end{aligned}
	\end{equation}
\end{theorem}

The proof of Theorem \ref{The logical operators of 4D XYZ product code} is given in the Appendix \ref{Proof of Theorem 2}.

Next, we will prove the upper bound of code distance of 4D XYZ product code in \textbf{corollary} \ref{code distance proof}. Before proving it, we first present \textbf{Lemma} \ref{minimum weight} which will be used later.

\begin{lemma}
	\label{minimum weight}
	Given a vector space $\mathcal{C}_{n}$ with dimension $n$ over $GF(2)$, two binary matrices $H_1$ and $H_2$ with size $n\times m_1$ and $m_2\times n$ respectively, the minimum weight of non-zero vector in $\mathcal{C}_{n}\backslash Im (H_1)$ and $\mathcal{C}_{n}\backslash \ker (H_2)$ is one.
\end{lemma}

The proof of \textbf{Lemma} \ref{minimum weight} is given in the Appendix \ref{Proof of Lemma 2}.

\begin{corollary}[\textbf{The upper bound of the code distance of 4D XYZ product code}]
	\label{code distance proof}
	The upper bound of code distance of 4D XYZ product code is
	\begin{equation}
		d\leq\min\{d_1, d_2, d_3, d_4\}
	\end{equation}
	where $d_1$, $d_2$, $d_3$ and $d_4$ is the minimum weight of the non-zero vectors in spaces $\ker\left(\begin{bmatrix}H_{x_1}\\H_{z_1}\end{bmatrix}\right)$, $\ker\left(\begin{bmatrix}H_{x_2}\\H_{z_2}\end{bmatrix}\right)$, $\ker\left(\left[H_{z_1}^T,H_{x_1}^T\right]\right)$ and $\ker\left(\left[H_{z_2}^T,H_{x_2}^T\right]\right)$, respectively.
\end{corollary}

\begin{proof}
	As proved in \textbf{Theorem} \ref{The logical operators of 4D XYZ product code}, the logical operators of 4D XYZ product code can be divided into two types. Thus, we should consider the upper bound of the minimum weight of these two types of  logical operators separately.
	
	First considering the first type of logical operators, and let $\alpha$ and $\beta$ (whose weight is $d_\beta$) be minimal weight non-zero vectors in the spaces $\ker\left(\begin{bmatrix}H_{x_1}\\H_{z_1}\end{bmatrix}\right)$ and $\left(\mathcal{C}_{n_B}\backslash Im\left(\left[H_{x_2}^T,H_{z_2}^T\right]\right)\right)$ respectively, then we have $\alpha\otimes \beta\in \ker\left(\begin{bmatrix}H_{x_1}\\H_{z_1}\end{bmatrix}\right)\otimes \left(\mathcal{C}_{n_B}\backslash Im\left(\left[H_{x_2}^T,H_{z_2}^T\right]\right)\right)$ and the weight of $\alpha\otimes \beta$ is $d_1d_\beta$. According to \textbf{Lemma} \ref{minimum weight}, we have $d_\beta=1$. Thus the weight of $\alpha\otimes \beta$ is $d_1$. For $\left(\mathcal{C}_{n_A}\backslash Im\left(\left[H_{x_1}^T,H_{z_1}^T\right]\right)\right) \otimes \ker\left(\begin{bmatrix}H_{x_2}\\H_{z_2}\end{bmatrix}\right)$, we have similar conclusion. Thus the upper bound of minimum weight of the vectors in space $\ker\left(\begin{bmatrix}H_{x_1}\\H_{z_1}\end{bmatrix}\right)\otimes\left(\mathcal{C}_{n_B}\backslash Im\left(\left[H_{x_2}^T,H_{z_2}^T\right]\right)\right) $ and $\left(\mathcal{C}_{n_A}\backslash Im\left(\left[H_{x_1}^T,H_{z_1}^T\right]\right)\right) \otimes \ker\left(\begin{bmatrix}H_{x_2}\\H_{z_2}\end{bmatrix}\right)$ is $d_1$ and $d_2$, respectively.
	
	For the second type of logical operators, the upper bound of the minimum weight for the vector ${\pi\left(\textbf{\emph{a}}_1^T,\textbf{\emph{a}}_2^T,\textbf{\emph{b}}_1^T,\textbf{\emph{b}}_2^T\right)}^T\in\ker{\left(\left[H_{z_1}^T,H_{x_1}^T\right]\right)}\otimes \left(\mathcal{C}_{m_3+m_4}\backslash Im\begin{bmatrix}
		H_{x_2}\\H_{z_2}
	\end{bmatrix}\right)$ is $d_3$. Observing Eq. (\ref{second logical operators}), it is obvious that the minimum weight of the logical operators $X_{L_2}$ and $X_{L_3}$ is no more than $d_3$. Similarly, the the minimum weight of the logical operators $Z_{L_2}$ and $Z_{L_3}$ is no more than $d_4$.
	
	To sum up, the upper bound of code distance is no more than the minimum of $d_1$, $d_2$, $d_3$ and $d_4$.
\end{proof}

Here we exploit Monte Carlo method proposed in Ref. \cite{liang2024determining,bravyi2024high}
to verify the code distance of the 4D Chamon code and the 4D XYZ product concatenated code. The corresponding results are consistent with \textbf{Corollary} \ref{code distance proof}. Further, we compare the code dimension and code distance between the 4D Chamon code, the 4D toric code, the 4D XYZ product concatenated code and the 4D homological product concatenated code in Table \ref{Comparison}, where $n_1, n_2, n_3, n_4$ are the code length of the repetition codes. The comparison results support that using two identical component CSS codes, the 4D XYZ product can construct non-CSS codes with higher code dimension or code distance than the CSS codes constructed by the 4D homological product.
\begin{table*}[htbp]
	\begin{center}
		\caption{Comparison of 4D Chamon codes, 4D toric codes, 4D XYZ product concatenated codes and 4D homological product concatenated codes in code dimension and code distance.}		
		\begin{tabular}{c|c|c|c}
			\hline
			& $n_1$, $n_2$, $n_3$, $n_4$ & code dimension & code distance \\
			\hline
			\multirow{7}{*}{4D Chamon codes} & $2,2,2,2$ & 32 & 4 \\
			\cline{2-4}
			& $3,3,3,3$ & 72 & 6 \\
			\cline{2-4}
			& $4,4,4,4$ & 128 & 8 \\
			\cline{2-4}
			& $5,5,5,5$ & 200 & 10 \\
			\cline{2-4}
			& $2,3,2,3$ & 8 & 6 \\
			\cline{2-4}	
			& $3,4,3,4$ & 8 & 12 \\
			\cline{2-4}
			& $4,5,4,5$ & 8 & 20 \\			
			\hline
			\multirow{7}{*}{4D toric codes} & $2,2,2,2$ & 6 & 4 \\
			\cline{2-4}
			& $3,3,3,3$ & 6 & 9 \\
			\cline{2-4}
			& $4,4,4,4$ & 6 & 16 \\
			\cline{2-4}
			& $5,5,5,5$ & 6 & 25 \\
			\cline{2-4}
			& $2,3,2,3$ & 6 & 4 \\
			\cline{2-4}	
			& $3,4,3,4$ & 6 & 9 \\
			\cline{2-4}
			& $4,5,4,5$ & 6 & 16 \\			
			\hline
			\multirow{5}{*}{4D XYZ product} & $3,3,3,3$ & 1 & 9 \\
			\cline{2-4}
			\multirow{5}{*}{concatenated codes}& $5,5,5,5$ & 1 & 25 \\
			\cline{2-4}
			& $7,7,7,7$ & 1 & 49 \\
			\cline{2-4}
			& $3,5,3,5$ & 1 & 15 \\
			\cline{2-4}
			& $3,7,3,7$ & 1 & 21 \\
			\hline
			\multirow{5}{*}{4D homological product} & $3,3,3,3$ & 1 & 9 \\
			\cline{2-4}
			\multirow{5}{*}{concatenated codes}& $5,5,5,5$ & 1 & 25 \\
			\cline{2-4}
			& $7,7,7,7$ & 1 & 49 \\
			\cline{2-4}
			& $3,5,3,5$ & 1 & 9 \\
			\cline{2-4}
			& $3,7,3,7$ & 1 & 9 \\
			\hline
		\end{tabular}
		\label{Comparison}
	\end{center}
\end{table*}

\section {Numerical simulations}
\label{4}
In this section, we exploit FDBP-OSD-0 proposed in Ref. \cite{yi2025improved} to study the error-correcting performance of the 3D Chamon code, the 4D Chamon code and the 4D XYZ product concatenated code under the independent single-qubit Pauli noise model, in which each qubit independently suffers a Pauli error $I$, $X$, $Y$, or $Z$ with error probability $1-p$, $p_x$, $p_y$, $p_z$, respectively, where $p=p_x+p_y+p_z$ is physical qubit error rate. Under this noise model, the bias rate of Pauli $Z$ noise is defined as
\begin{equation}
	\label{eta}
	\eta=\frac{p_z}{p_x+p_y}
\end{equation}
For instance, the bias rate of the depolarizing and pure Pauli $Z$ noise is $\eta=0.5$ and $\eta=\infty$, respectively. In this paper, we only consider noiseless stabilizer measurements.

It is important to note that this paper does not focus on optimizing the decoding algorithm for 3D and 4D XYZ product codes. Instead, we use the FDBP-OSD-0 algorithm, which is an improved version of binary BP and is applicable to all quantum stabilizer codes with improved performance over traditional binary BP. \cite{babar2015fifteen,roffe2020decoding}. However, as shown in Ref. \cite{PRXQuantum.4.020332}, even though tradition binary BP combined with OSD \cite{Panteleev2021degeneratequantum}  has demonstrated strong decoding performance for a wide range of QLDPC codes, but its performance can degrade for the codes with large code length and high degeneracy. FDBP also has the same problem, and its performance is also greatly degraded by short cycles, especially 4-cycles in the Tanner graph. In Appendix \ref{number of 4-cycles}, we compare the number of 4-cycles in the Tanner graphs of the 3D Chamon code, the 3D toric code, the 4D Chamon code and the 4D toric code. Our analysis reveals that 4D XYZ product codes exhibit significantly more 4-cycles in their Tanner graphs compared to 3D and 4D homological product codes. Thus, the error-correcting performance of the 3D Chamon code, the 4D Chamon code and the 4D XYZ product concatenated code obtained by FDBP-OSD-0 in this paper is not the optimal. Nonetheless, from the perspective of code-capacity threshold, they still show good error-correcting performance against Pauli-$Z$-biased noise.

\subsection {The 3D Chamon code}
\label{3D Chamon code}
The 3D Chamon code proposed in Ref. \cite{PhysRevLett.94.040402} is an instance of the 3D XYZ product of three repetition codes \cite{leverrier2022quantum}. Its characteristics has been studied in depth by Bravyi \emph{et al} \cite{BRAVYI2011839}. In Ref. \cite{schwartzman2025generalizing} and \cite{zhao2023quantum}, the researchers study its error-correcting performance (namely, code-capacity threshold) against the depolarizing noise and pure Pauli $X$ noise. However, the corresponding code-capacity thresholds obtained in these work is not satisfying enough. Our work aims to further explore the performance under various noise models.

The code dimension of the 3D Chamon code is $4\gcd\left(n_1,n_2,n_3\right)$, where $n_1$, $n_2$ and $n_3$ are the code length of three repetition codes, respectively. Besides, we can also understand 3D Chamon code from the perspective of three-dimensional layout. As shown in Fig. \ref{Chamoncode}, on a three-dimensional cubic Bravais lattice, the qubits of the 3D Chamon code are placed on the face centers each of which neighbors two cubes and on the vertices each of which is shared in eight cubes. Each edge of cube represents a stabilizer that involves neighboring six qubits. These six qubits are divided into three classes, each of which contains two qubits along one of the $x$, $y$, $z$ three directions. For each class of qubits, the type of the Pauli operators that the stabilizer acts on them is consistent with the type of direction. We conjecture that the 3D Chamon code may exhibit identical error-correcting capability for Pauli $X$, $Z$ and $Y$ errors when $n_1=n_2=n_3$, due to its spatial symmetry in this configuration.
\begin{figure}[htbp]
	\centering
	\includegraphics[width=0.4\textwidth]{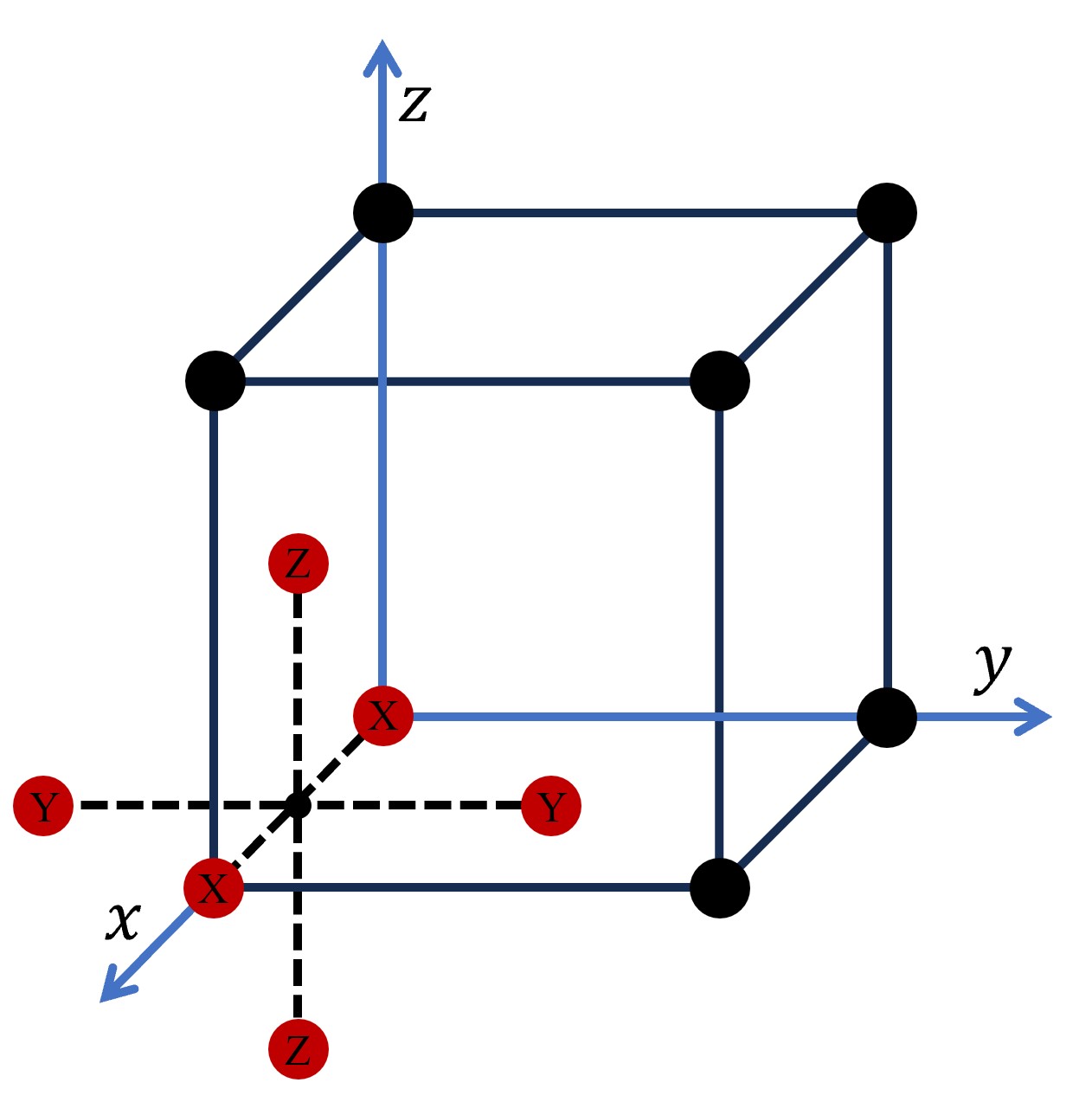}
	\caption{Three-dimensional layout of the 3D Chamon code. Qubits are placed on the face centers and the vertices. The geometric centers of edges and cubes represent stabilizers, which involve neighboring six qubits. Some qubits are omitted for clearly.}
	\label{Chamoncode}
\end{figure}

In Fig. \ref{3DChamonCodeCut}, we consider the case of $n_1=n_2=n_3=n$, and the corresponding 3D Chamon code encodes $4n$ logical qubits. Fig. \ref{3DChamonCodeCut}(a)$\sim$(c) show the total logical error rate versus the physical qubit error rate $p$ under the pure Pauli $X$, $Y$ and $Z$ noise, respectively. One can see that the code-capacity thresholds of the 3D Chamon code under the pure Pauli $X$, $Y$ and $Z$ noise model are very close, which are all around $14.5\%$ respectively, and much higher than the result, $4.92\%$, in Ref. \cite{zhao2023quantum}. These results support our conjecture. Specifically, the 3D Chamon code achieves equivalent error-correction capability for Pauli $X$, $Z$, and $Y$ errors when $n_1=n_2=n_3$. Fig. \ref{3DChamonCodeCut}(d) shows that the code-capacity threshold of the 3D Chamon code under depolarizing noise model is around $14.5\%$, which is also much higher than the result, $10.5\%$, in Ref. \cite{schwartzman2025generalizing}. In the legend, the notation $L=n_1,n_2,n_3$ represents the corresponding 3D Chamon code, which is constructed by three repetition codes with code length of $n_1$, $n_2$ and $n_3$.
\begin{figure}[htbp]
	\centering
	\includegraphics[width=0.49\textwidth]{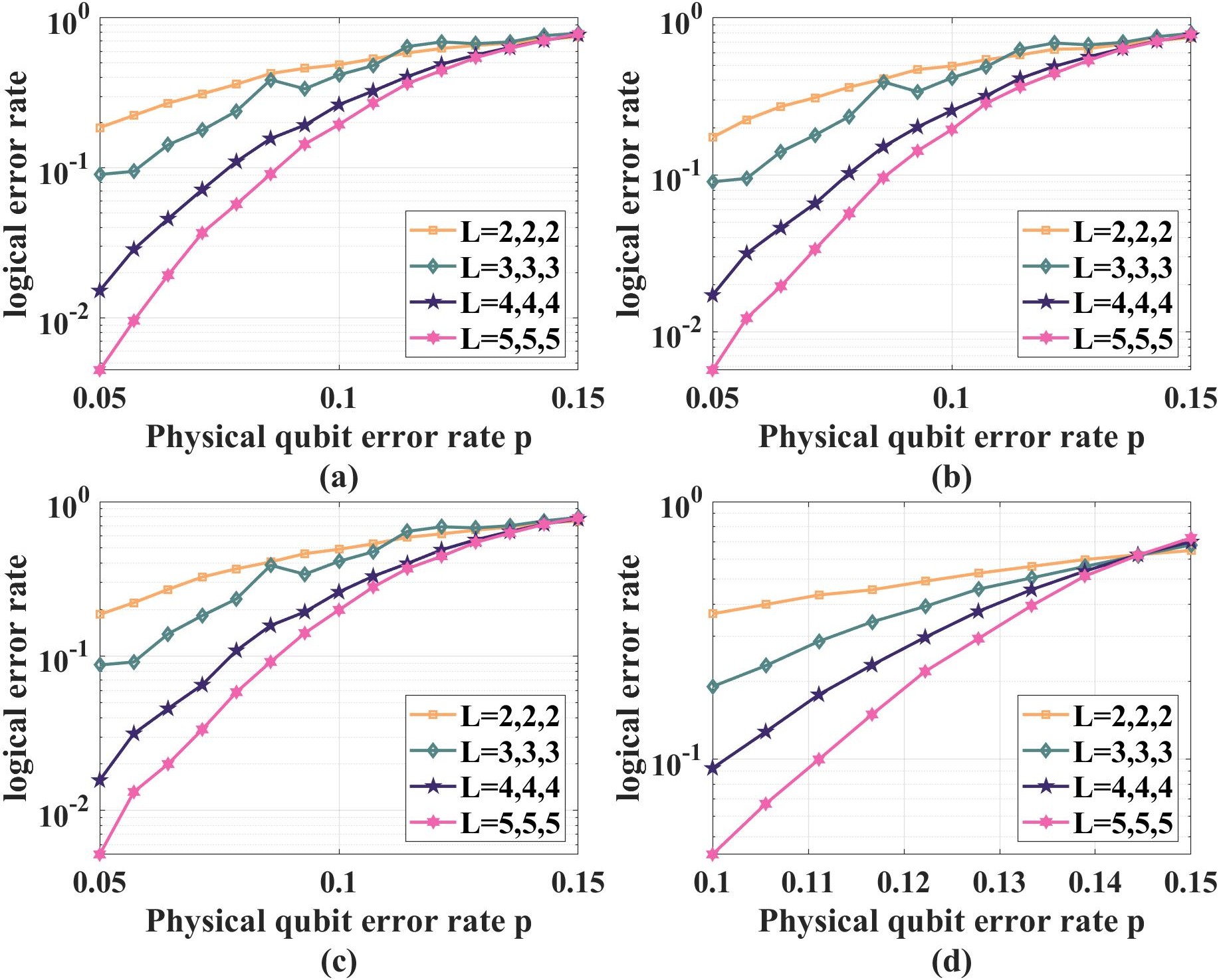}
	\caption{The error-correcting performance of the Chamon code under (a) the pure Pauli $X$ noise model, (b) the pure Pauli $Z$ noise model, (c) the pure Pauli $Y$ noise model and (d) the depolarizing noise model.}
	\label{3DChamonCodeCut}
\end{figure}

\subsection {The 4D Chamon code}
\label{4D Chamon code}
We refer to the code which is the 4D XYZ product of four repetition codes with code length $n_1$, $n_2$, $n_3$ and $n_4$ as the 4D Chamon code and conjecture that it is the 4D generalization of the 3D Chamon code, just as the relation between the 3D and 4D toric codes. This conjecture rests on two key observations. First, the code dimension of the 4D Chamon code, $8\gcd(n_1,n_2)\gcd(n_3,n_4)$, is similar with that of the 3D Chamon code, $4\gcd\left(n_1,n_2,n_3\right)$. Second, from the perspective of code construction, as shown in Fig. \ref{XYZchaincomplex} and \ref{4DXYZproduct}, the tensor structures of the 3D and 4D XYZ product are the same as those of the 3D and 4D homological product, respectively. Since the 3D and 4D toric codes are constructed from the 3D and 4D homological product, this structural correspondence reinforces the generalization relationship.

In Fig. \ref{FourDXYZCompareHPtoricDepCut}, we first consider the case of $n_1=n_2=n_3=n_4$. In this case, the 4D Chamon code encodes $8n_1^2$ logical qubits, while the 4D toric code always encodes 6 logical qubits. In the legend, the notation $L=n_1,n_2,n_3,n_4$ represents the corresponding code, which is constructed by four repetition codes with code length of $n_1$, $n_2$, $n_3$ and $n_4$. Fig. \ref{FourDXYZCompareHPtoricDepCut}(a) and (b) show the total logical error rate versus the physical qubit error rate $p$ of the 4D Chamon code and the 4D toric code under the depolarizing noise model, respectively. One can observe that the code-capacity thresholds of the 4D Chamon code and 4D toric code obtained by FDBP-OSD-0 are around $13\%$ and $16\%$, respectively. Under the pure Pauli $Z$ noise, Fig. \ref{FourDXYZCompareHPtoricDepCut}(c) and (d) show that the thresholds of the 4D Chamon code and the 4D toric code obtained by FDBP-OSD-0 are around $13\%$ and $10.5\%$, respectively. The simulation results exhibit that the 4D Chamon code possesses comparable error-correcting performance to the 4D toric code while encoding more logical qubits, thus demonstrating the superior error-correction capability of the 4D XYZ product code construction.
\begin{figure}
	\centering
	\includegraphics[width=0.49\textwidth]{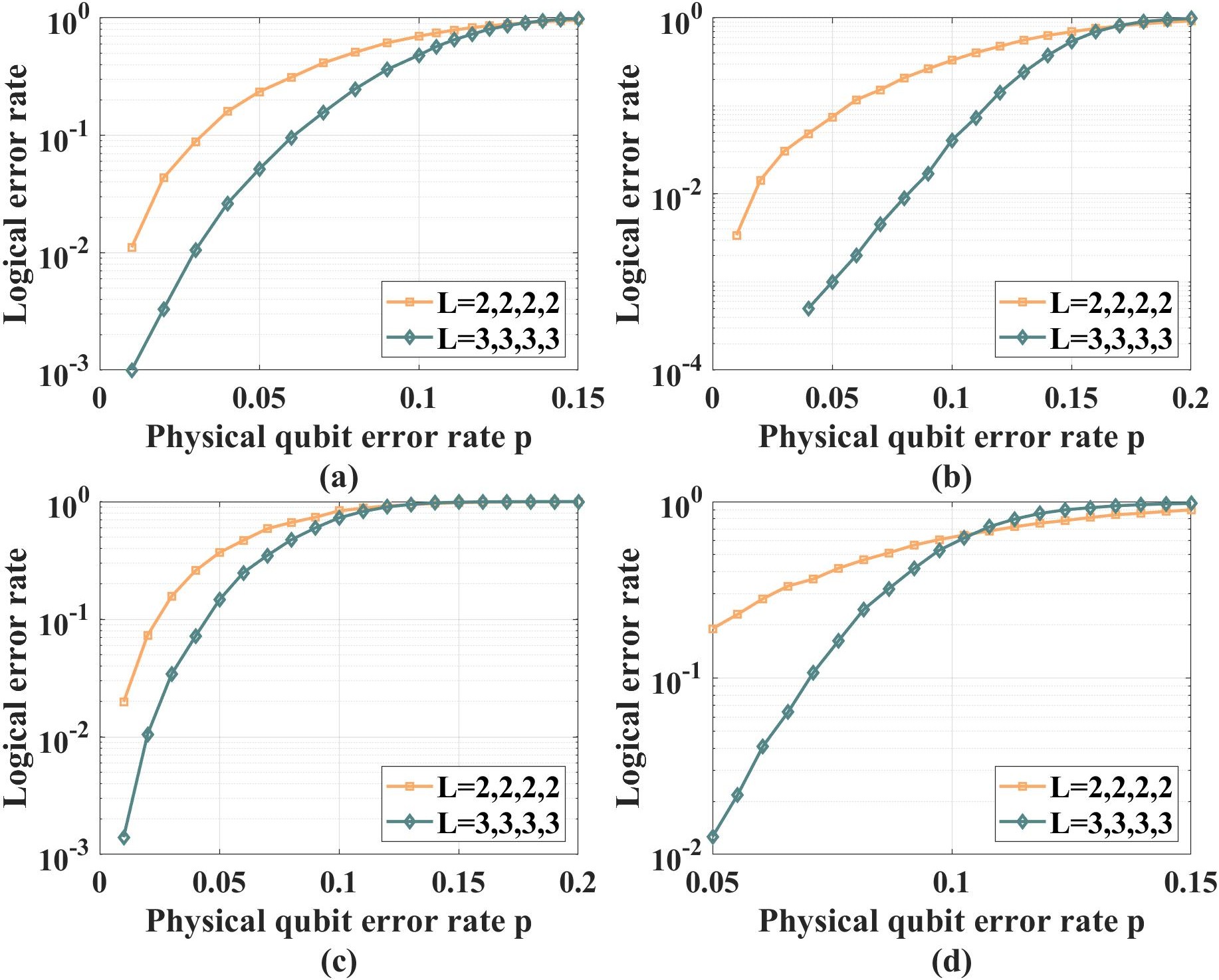}
	\caption{When $n_1=n_2=n_3=n_4$, the error-correcting performance of (a) the 4D Chamon code under the depolarizing noise, (b) the 4D toric code under the depolarizing noise, (c) the 4D Chamon code under the pure Pauli $Z$ noise and (d) the 4D toric code under the pure Pauli $Z$ noise.}
	\label{FourDXYZCompareHPtoricDepCut}
\end{figure}

We second consider the case that $n_1$, $n_2$, $n_3$ and $n_4$ are coprime. In this case, the corresponding 4D Chamon code always encode $8$ logical qubits, which is close to the 4D toric code. Thus, it is more reasonable to compare the error-correcting performance of the 4D Chamon code with the 4D toric code in this case. First, we consider the pure Pauli $Z$ noise. As shown in Fig. \ref{FourDXYZCompareHPToricPureZCoprimeCut} (a) and (b), the code capacity thresholds of the 4D Chamon code and the 4D toric code obtained by FDBP-OSD-0 under pure Pauli $Z$ noise are around $19\%$ and $10.5\%$, respectively. These results implies that the 4D Chamon code possess better error-correcting performance than that of the 4D toric code against the pure Pauli $Z$ noise.
\begin{figure}
	\centering
	\includegraphics[width=0.49\textwidth]{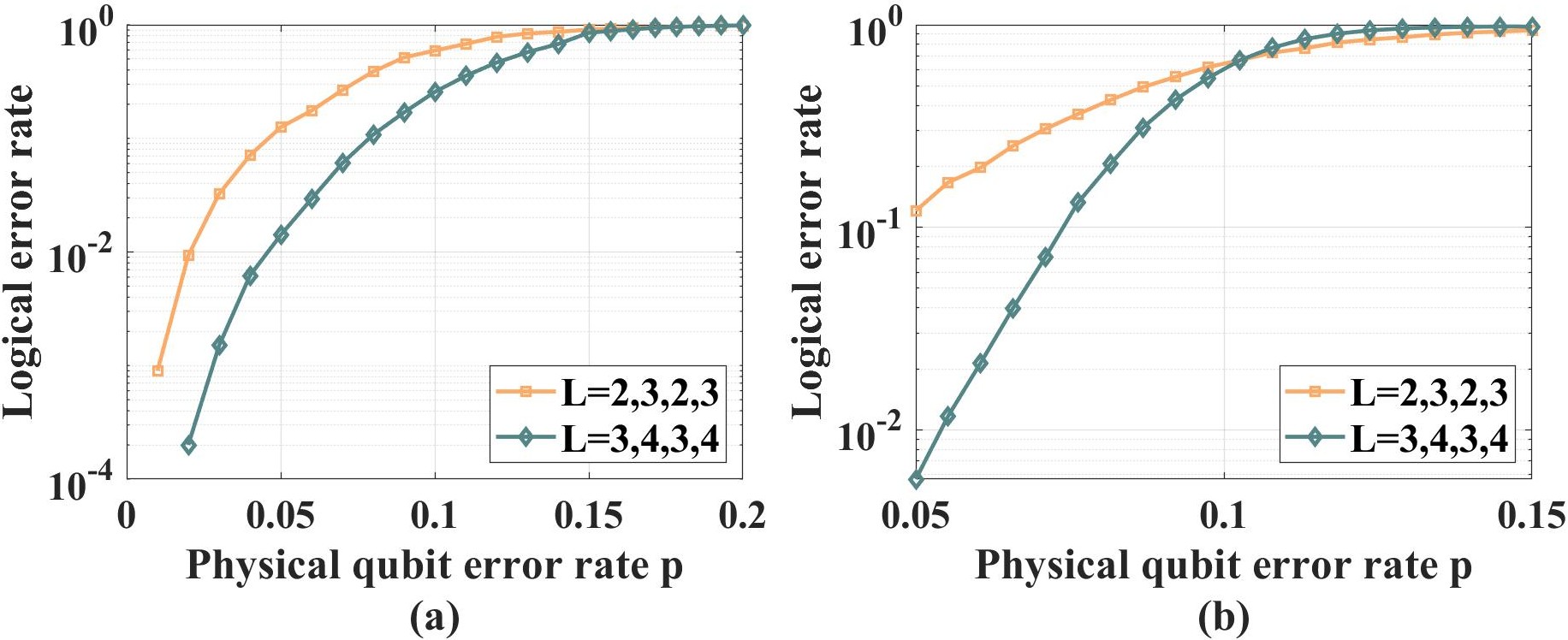}
	\caption{When $n_1$, $n_2$, $n_3$ and $n_4$ are coprime, the error-correcting performance of (a) the 4D Chamon code and (b) the 4D toric code against the pure Pauli $Z$ noise.}
	\label{FourDXYZCompareHPToricPureZCoprimeCut}
\end{figure}

Further, we explore the error-correcting performance of the 4D Chamon code under Pauli-$Z$-biased noise with bias rate $\eta=1000$, $100$, $10$ and $0.5$. Our simulation results show that the corresponding code-capacity thresholds are around $18.5\%$, $18\%$, $17\%$ and $13\%$, respectively. These results imply that, the 4D Chamon code encoding $8$ logical qubits possesses good error-correcting performance against Pauli-$Z$-biased noise.
\begin{figure}[htbp]
	\centering
	\includegraphics[width=0.49\textwidth]{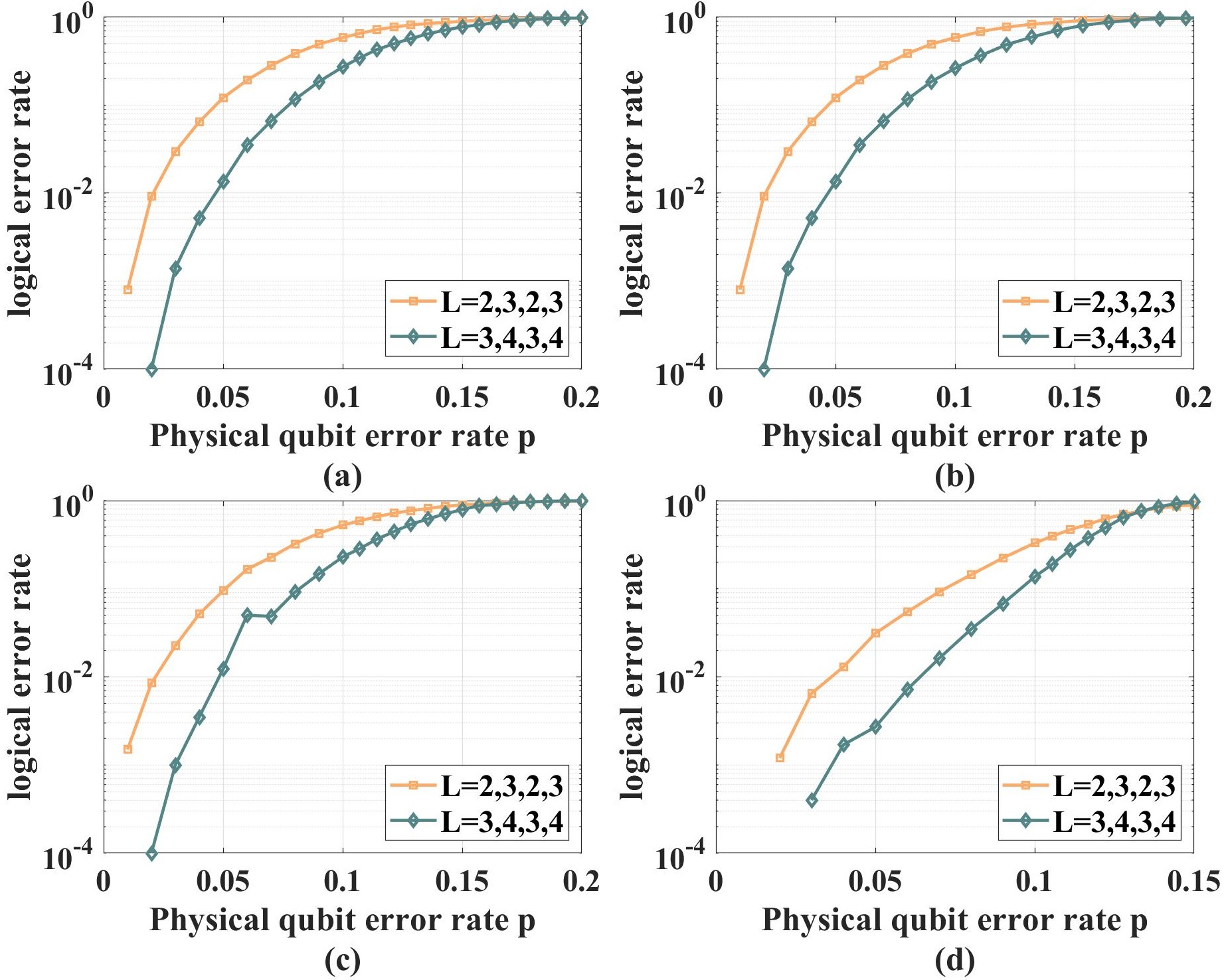}
	\caption{When $n_1$, $n_2$, $n_3$ and $n_4$ are coprime, The error-correcting performance of the 4D Chamon codes against Pauli-$Z$-biased noise with bias rate (a) $\eta=1000$, (b) $100$, (c) $10$ and (d) $0.5$.}
	\label{FourDXYZCompareHPToricZBiasCut}
\end{figure}

\subsection {The 4D XYZ product concatenated code}
\label{4D XYZ concatenated code}
As described in Sect. \ref{Dimension}, the 4D XYZ product concatenated code is constructed from two concatenated codes, which are obtained from two pairs of repetition codes with block lengths $\left(n_1,\ n_2\right)$ and $\left(n_3,\ n_4\right)$ ($n_1$, $n_2$, $n_3$ and $n_4$ are all odd), and only encodes one logical qubit. According to \textbf{Theorem} \ref{The logical operators of 4D XYZ product code} and \textbf{Corollary} \ref{code distance proof}, the minimum weight of its logical $Z$ operator is $n_3n_4$, thus we conjecture that increasing the block lengths $\left(n_3,\ n_4\right)$ of the pair of repetition codes can construct the corresponding 4D XYZ product concatenated codes, which have better error-correcting performance against Pauli-$Z$-biased noise than 4D homological product concatenated codes.

First, we consider the pure Pauli $Z$ noise. In Fig. \ref{FourDXYZCompareHPShorPureZCut}(a), the code-capacity threshold of the 4D XYZ product concatenated code obtained by FDBP-OSD-0 is around $38\%$, which is much higher than that of the 4D homological product concatenated code, which is around 15$\%$ as shown in Fig. \ref{FourDXYZCompareHPShorPureZCut}(b).

\begin{figure}[]
	\centering
	\includegraphics[width=0.49\textwidth]{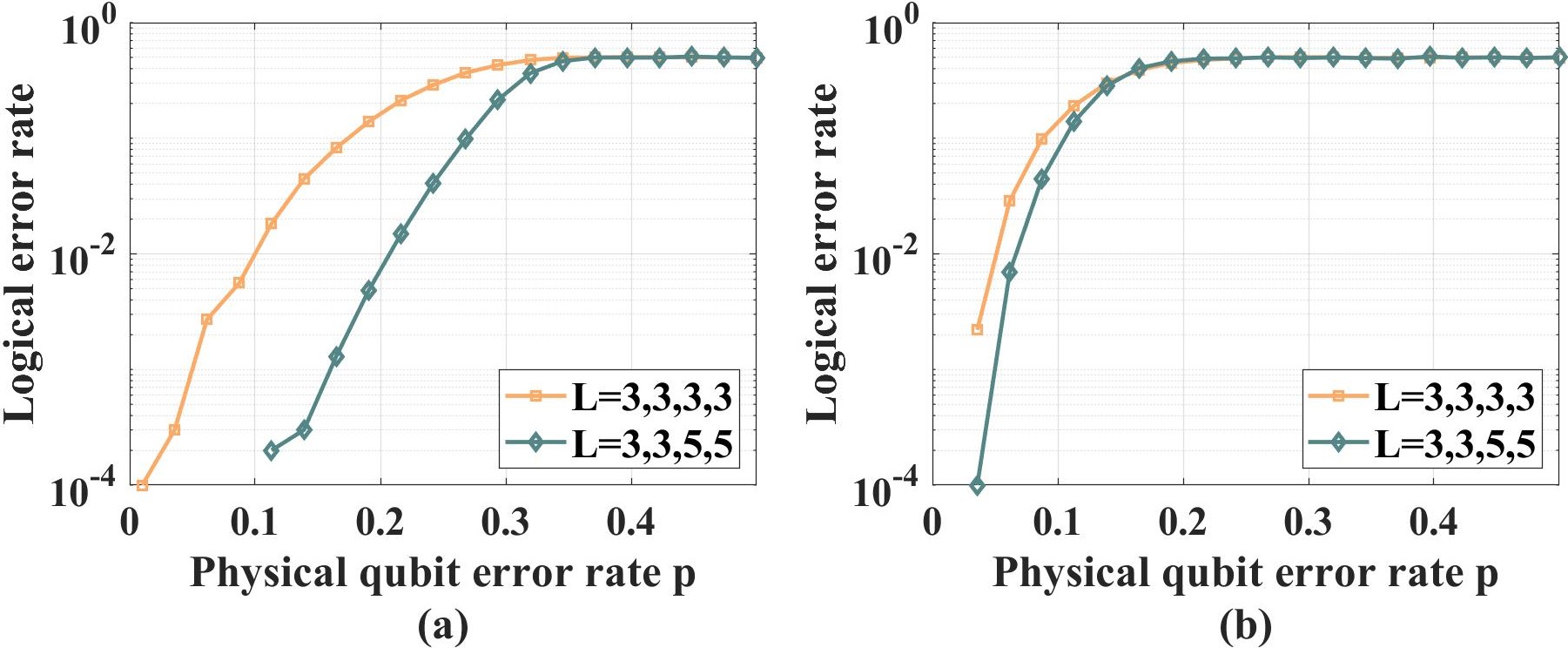}
	\caption{The error-correcting performance of (a) the 4D XYZ product concatenated code and (b) the 4D homological product concatenated code against the pure Pauli $Z$ noise.}
	\label{FourDXYZCompareHPShorPureZCut}
\end{figure}

Further, we explore the error-correcting performance of the 4D XYZ product concatenated code against Pauli-$Z$-biased noise with bias rate $\eta=1000$, 100, 10 and 0.5. Our simulation results show that the corresponding code-capacity thresholds are around $37\%$, $32\%$, $22\%$ and $10\%$, respectively, which support that the 4D XYZ product can construct codes with good error-correcting performance for Pauli-$Z$-biased noise.
\begin{figure}[htbp]
	\centering
	\includegraphics[width=0.49\textwidth]{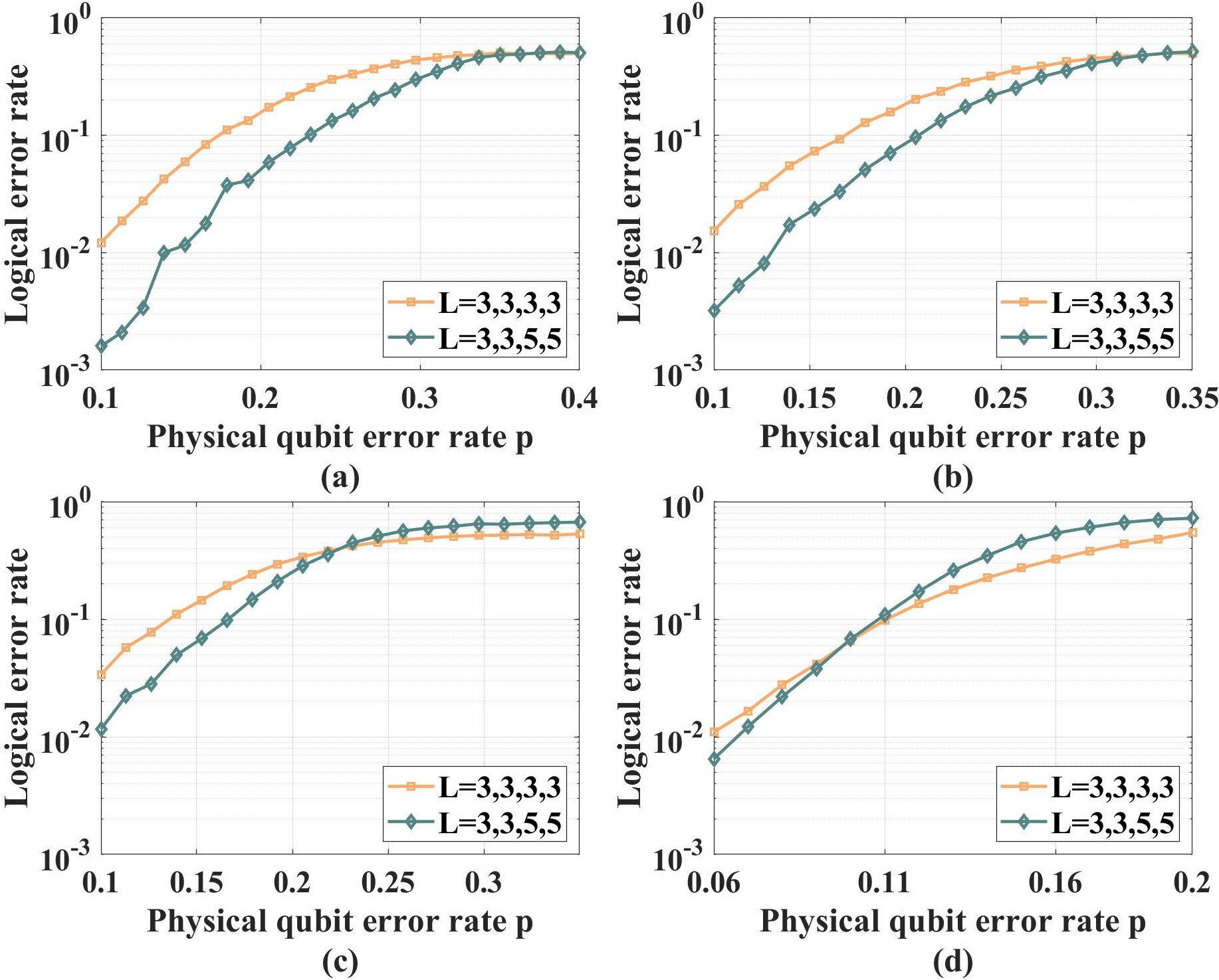}
	\caption{The error-correcting performance of the 4D XYZ product concatenated code against Pauli-$Z$-biased noise with bias rate (a) $\eta=1000$, (b) 100, (c) 10 and (d) 0.5.}
	\label{FourDXYZCompareHPShorZBiasCut}
\end{figure}

\section {Further discussions on the 4D Chamon code}
\label{Further discussions}
In Corollary \ref{dimension of the 4D Chamon code}, we prove that the code dimension of the 4D Chamon code is $8\gcd(n_1,n_2)\gcd(n_3,n_3)$, which is similar to the 3D Chamon code, both being dependent on cubic lattice size. This constitutes one of the salient features inherent to fracton models, while the 3D Chamon code has been rigorously proved to be a 3D fracton model \cite{BRAVYI2011839}. Thus, this naturally prompts a question: is the 4D Chamon code a 4D fracton model? In this paper, we do not rigorously address this question. However, we identify that the 4D Chamon code possesses another salient feature of fracton models—the restricted mobility of excitations. Consequently, we strongly suggest that the 4D Chamon code is likely a novel 4D fracton model \cite{shen2022fracton, li2021fracton}.

To demonstrate that the 4D Chamon code possesses the restricted mobility of excitations, we first specify its geometric arrangement—it is an spin model with eight-qubit nearest-neighbor interactions on a 4D cubic lattice with periodic boundary conditions. We reveal that it possesses rigid logical operators, which corresponds to the restricted mobility of excitations \cite{huang2023tailoring}. 

Considering a 4D cubic lattice $\Lambda=\mathbb{Z}_{2L_x}\times\mathbb{Z}_{2L_y}\times\mathbb{Z}_{2L_z}\times\mathbb{Z}_{2L_w}$ with periodic boundary conditions and linear dimensions $2L_x$, $2L_y$, $2L_z$, $2L_w$. The factor $2$ comes from that we set the lattice constant to be 2 for convenience. Adopting the geometric notations introducing in Ref. \cite{li2021fracton}, any $d$-cube ($0\leq d\leq 4$ dimensional object, namely, 0-cube is vertex, 1-cube is edge, 2-cube is face, 3-cube is cube and 4-cube is hypercube) in $\Lambda$ is denoted by its coordinate of geometric center.  The coordinate $u=(i,j,k,l)\in\Lambda$ of a $d$-cube always contains $(4-d)$ even numbers and $d$ odd numbers. For instance, $(0,0,0,0)$ represents a 0-cube (vertex); $(1,0,0,0)$ represents a 1-cube (edge) whose center is $(1,0,0,0)$; $(1,1,0,0)$ represents a 2-cube (face) whose center is $(1,1,0,0)$; $(1,1,1,0)$ represents a 3-cube (cube) whose center is $(1,1,1,0)$; $(1,1,1,1)$ represents a 4-cube (hypercube) whose center is $(1,1,1,1)$.

The qubits of the 4D Chamon code are placed on the vertices, face centers and hypercube centers in $\Lambda$, and the cubes and edges in $\Lambda$ corresponds to stabilizer generators, which are divided in to four classes, $S$, $T$, $U$ and $V$. Formally,

1. For stabilizer $g_S$ in $S$, it corresponds to a cube $\gamma_S$ along $x$ or $y$ axis. A cube is the common “$face$” shared by two nearest hypercubes. Thus, there are eight qubits nearest to a cube, namely, six qubits on the center of faces and two qubits on the center of hypercubes. For $g_S$, Pauli $Z$ operators act on two qubits placed on the $zw$ plane, Pauli $X$ operators act on four qubits placed on the other two types of planes and Pauli $Y$ operators act on two qubits placed on two centers of hypercubes, which share the cube $\gamma_S$. 

2. For stabilizer $g_T$ in $T$, it corresponds to a cube $\gamma_T$ along $z$ or $w$ axis, and Pauli $Y$ operators act on two qubits placed on the $xy$ plane, Pauli $Z$ operators act on four qubits placed on the other two types of planes and Pauli $X$ operators act on two qubits placed on two centers of hypercubes sharing the cube $\gamma_T$.

3. For stabilizer $g_U$ in $U$, it corresponds to an edge $e_U$ along $x$ or $y$ axis. An edge connects two vertices and is shared by six faces along three distinct axes. Thus, there are eight qubits nearest to an edge, namely, six qubits on the center of faces and two qubits on vertices. For $e_U$, Pauli $Z$ operators act on two qubits placed on the $xy$ plane, Pauli $X$ operators act on four qubits placed on the other two types of planes and Pauli $Y$ operators act on two qubits placed on two vertices connecting to edge $e_U$.

4. For stabilizer $g_V$ in $V$, it corresponds to an edge $e_V$ along $z$ or $w$ axis, and Pauli $Y$ operators act on two qubits placed on the $zw$ plane, Pauli $Z$ operators act on four qubits placed on the other two types of planes and Pauli $X$ operators act on two qubits placed on two vertices connecting to edge $e_V$.

Fig. \ref{4D Chamon code layout} shows two stabilizers in $S$ and $U$, respectively.

\begin{figure}[htbp]
	\centering
	\includegraphics[width=0.4\textwidth]{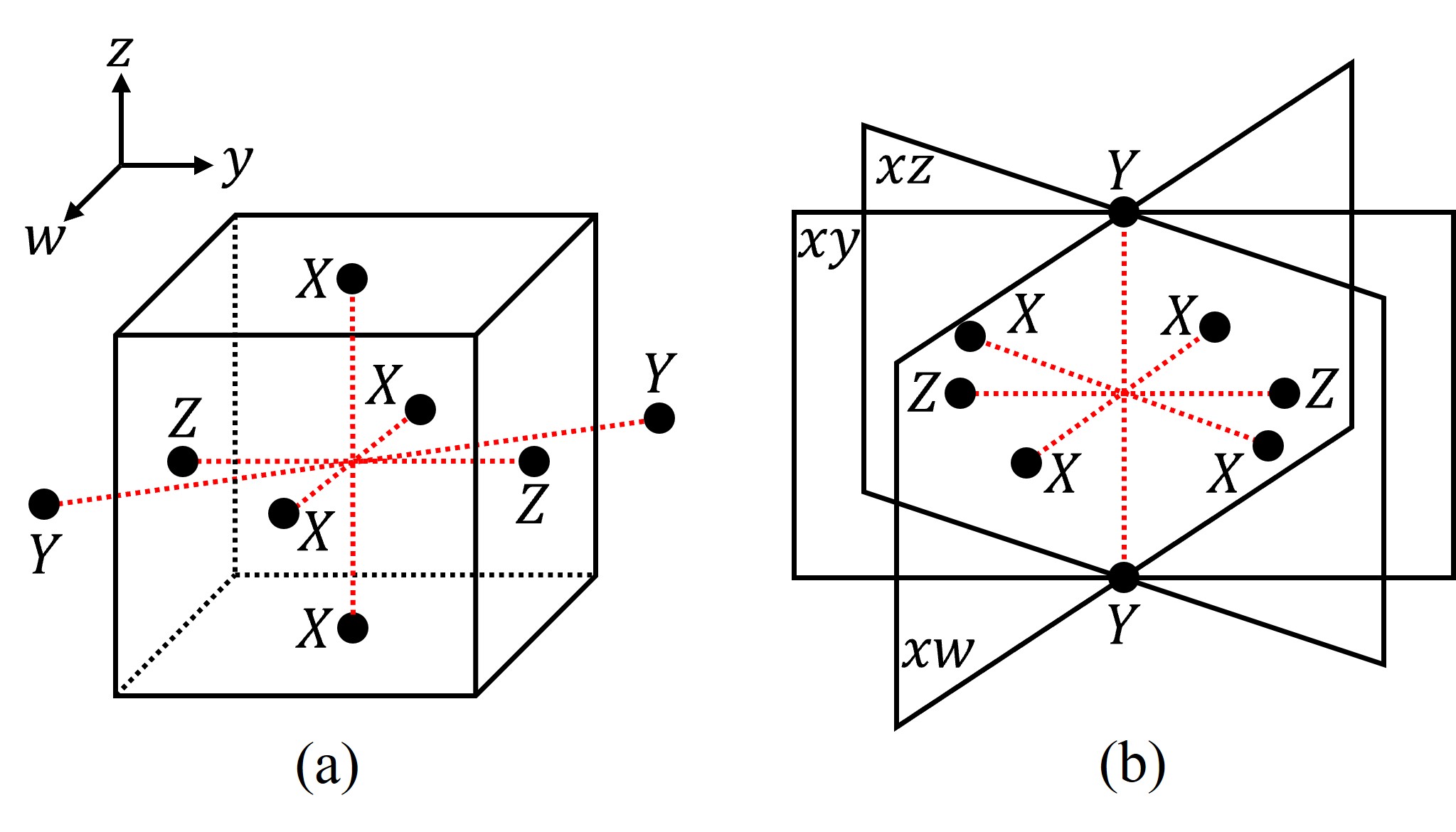}
	\caption{The stabilizers of the 4D Chamon code. (a) A stabilizer in $S$ which corresponds to a cube along $x$ axis. (b) A stabilizer in $U$ which corresponds to an edge along $x$ axis.}
	\label{4D Chamon code layout}
\end{figure}

Next, we first describe a class of string-like operators on a 4D cubic lattice, which are straight lines along the face-diagonals of $xy$ or $zw$ plane and create a pair of excitations, namely, a dipole, located near each end-point of the string. Second, we show that the string-like operators are rigid undeformable, thus the excitations are restricted to move along the face-diagonal, which is similar with 3D Chamon code \cite{BRAVYI2011839}. Considering an operator
\begin{equation}
	W(\gamma) = \prod\limits_{v\in V}X_v \prod\limits_{f\in F}Y_f
\end{equation}
where $V=\{(0,0,0,0), (2,2,0,0),\cdots,(2m,2m,0,0)\}$  is a set of vertices located along the face-diagonal $l$ of $xy$ plane, and $F=\{(1,1,0,0),(3,3,0,0),\cdots,(2m-1,2m-1,0,0)\}$ is a set of face whose centers also locate along $l$. One can check that $W(\gamma)$ commutes with all stabilizers in $S$ and $T$, and only anti-commutes with four stabilizers corresponding to four edges whose center coordinates are $(-1,0,0,0)$, $(0,-1,0,0)$, $(2m,2m+1,0,0)$ and $(2m+1,2m,0,0)$, respectively. An example is shown in Fig. \ref{RigidStringOperator}. One can check that the dipole can only move along the face-diagonal of $xy$ plane by applying Pauli $X$ operators on qubits located at vertices or Pauli $Y$ operators on qubits located at face centers, thus cannot freely move to other locations without creating new excitations. In addition, since the 4D cubic lattice under consideration possesses periodic boundary conditions, if $m = L_x = L_y = L$, the end-points of the operator $W(\gamma)$ are actually the same vertices, thus $W(\gamma)$ is a non-contractible loop which commutes with all stabilizers and is an logical operator corresponding to $X_{L_3}$ in Theorem\ref{The logical operators of 4D XYZ product code}. Thus, the 4D Chamon code possesses rigid logical operators corresponding to the restricted mobility of excitations. \cite{huang2023tailoring}.

\begin{figure}[htbp]
	\centering
	\includegraphics[width=0.4\textwidth]{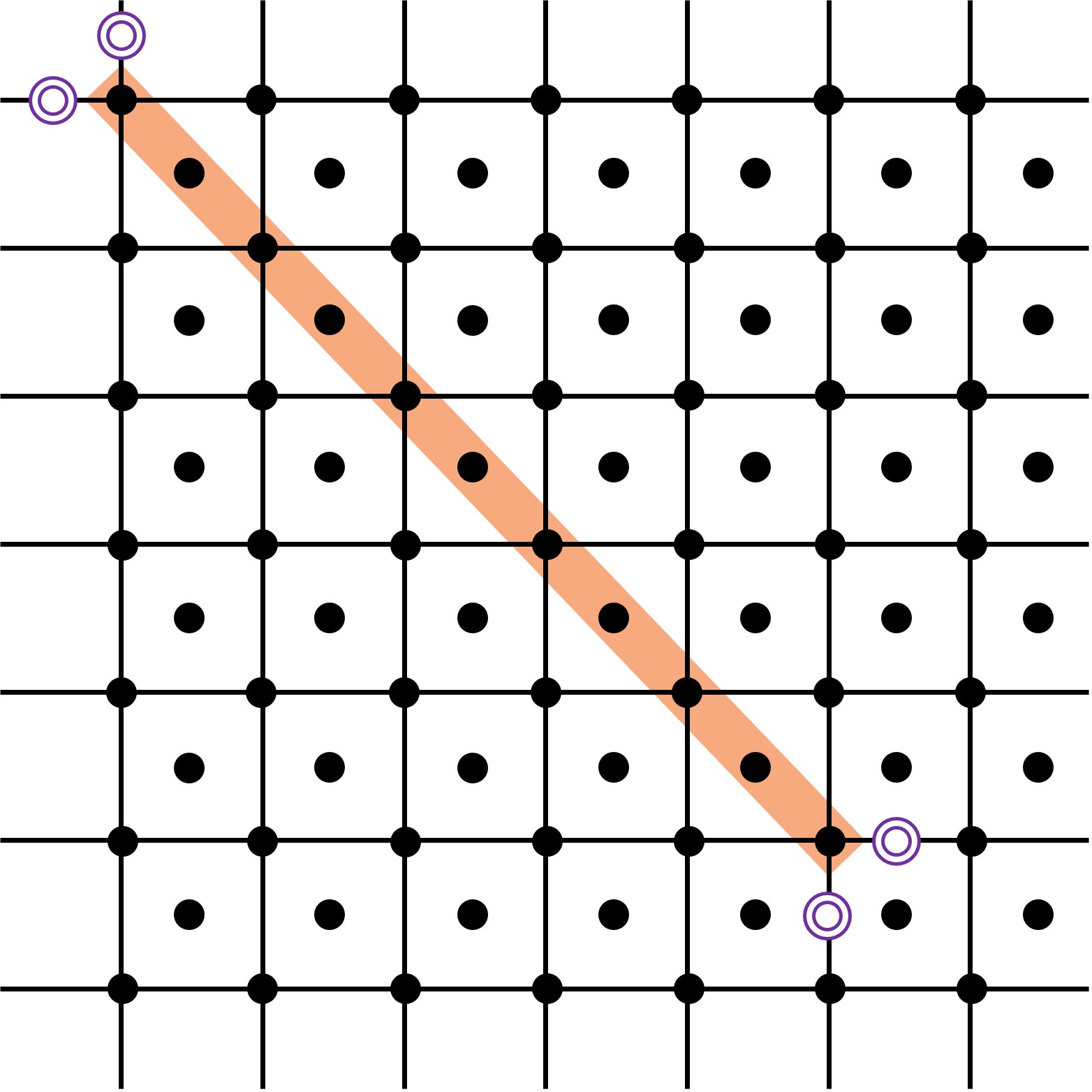}
	\caption{An example of rigid string operator $W(\gamma)$ locating at $xy$ plane and along the face-diagonal of $xy$ plane (the orange shaded
	region). Black dots represent qubits located at vertices and the center of faces. Double purple circles near the end-points of the string indicate locations of excitations created by $W(\gamma)$.  A pair of excitations located near each end-point of the string is called a dipole \cite{BRAVYI2011839}.}
	\label{RigidStringOperator}
\end{figure}

The above discussions strongly suggest that the 4D Chamon code is a novel 4D fracton model. This naturally leads to another question: what are the code parameters and error-correcting performance of the 4D Chamon code compared to the other 4D fracton model—the 4D X-cube model?

Table \ref{4DFractonModelComparison} compares the code dimension $k$ and code distance between the 4D Chamon code and 4D X-cube model. In both isotropic and anisotropic lattices, the code distance of the 4D Chamon code is larger than that of the 4D X-cube model. However, in isotropic lattices, the code dimension of the 4D Chamon code is significantly larger than that of the 4D X-cube model.
\begin{table*}
	\begin{center}
		\caption{Comparison on the code dimension $k$ and code distance between the 4D Chamon code and the 4D X-cube model}		
		\begin{tabular}{c|c|c|c}
			\hline
			& \multirow{2}{*}{$n_1$, $n_2$, $n_3$, $n_4$} & 4D Chamon code & 4D X-cube model \\
			&\ &$k=8\gcd(n_1,n_2)\gcd(n_3,n_4)$&$k=3\sum_{i=1}^{4}n_i - 6$ \\ 
			\hline
			\multirow{7}{*}{Code dimension $k$} & $2,2,2,2$ & 32 & 18 \\
			\cline{2-4}
			& $3,3,3,3$ & 72 & 30 \\
			\cline{2-4}
			& $4,4,4,4$ & 128 & 42 \\
			\cline{2-4}
			& $5,5,5,5$ & 200 & 54 \\
			\cline{2-4}
			& $2,3,2,3$ & 8 & 24 \\
			\cline{2-4}	
			& $3,4,3,4$ & 8 & 36 \\
			\cline{2-4}
			& $4,5,4,5$ & 8 & 48 \\			
			\hline
			\multirow{7}{*}{Code distance} & $2,2,2,2$ & 4 & 2 \\
			\cline{2-4}
			& $3,3,3,3$ & 6 & 3 \\
			\cline{2-4}
			& $4,4,4,4$ & 8 & 4 \\
			\cline{2-4}
			& $5,5,5,5$ & 10 & 5 \\
			\cline{2-4}
			& $2,3,2,3$ & 6 & 2 \\
			\cline{2-4}	
			& $3,4,3,4$ & 12 & 3 \\
			\cline{2-4}
			& $4,5,4,5$ & 20 & 4 \\			
			\hline
		\end{tabular}
		\label{4DFractonModelComparison}
	\end{center}
\end{table*}

Fig. \ref{FourDXCubeDepPureZ}(a) and (b) show the total logical error rate curves of the 4D X-cube model obtained by FDBP-OSD-0 under the depolarizing noise and pure Pauli $Z$ noise, and the intersections of curves are around $1.4\%$ and $1.2\%$, respectively. These values are much smaller than those values of the 4D Chamon code as shown in Fig. \ref{FourDXYZCompareHPtoricDepCut}.

\begin{figure}
	\centering
	\includegraphics[width=0.49\textwidth]{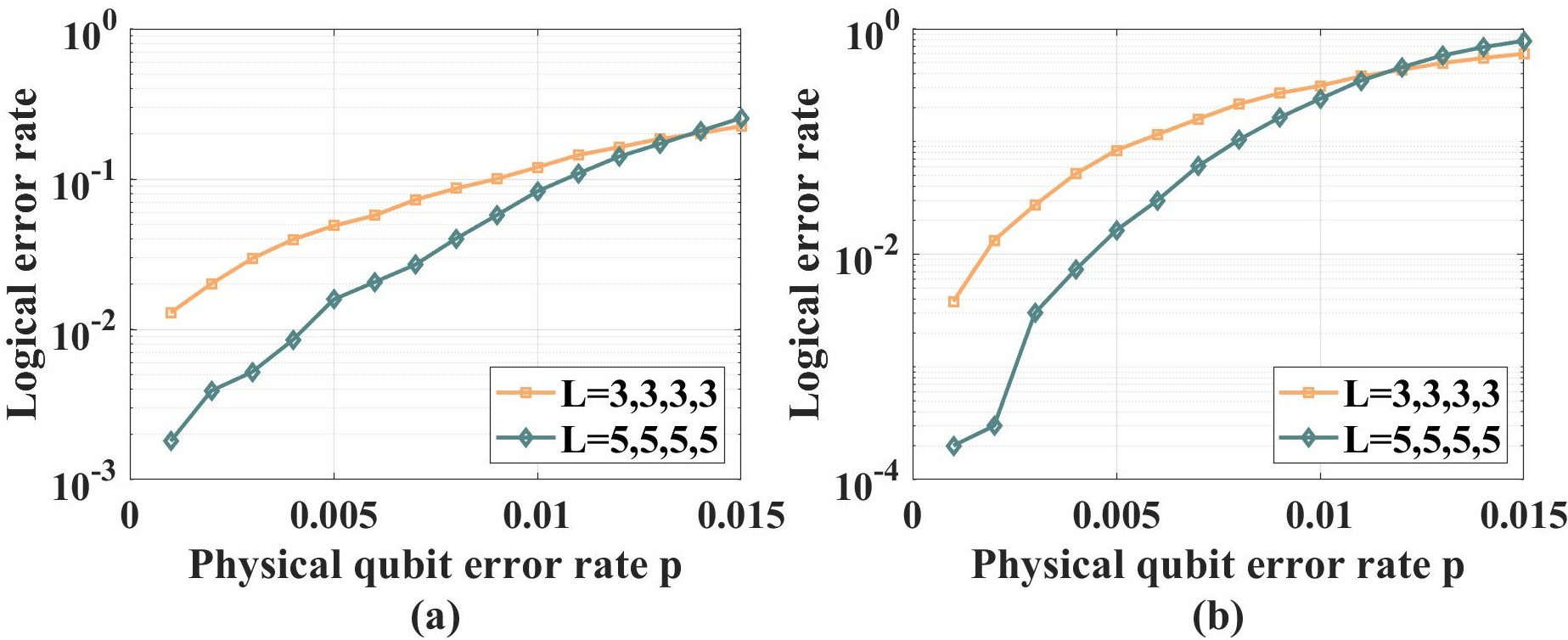}
	\caption{The error-correcting performance of 4D X-cube model against (a) depolarizing noise and (b) pure Pauli $Z$ noise.}
	\label{FourDXCubeDepPureZ}
\end{figure}

\section {Conclusion}
\label{6}
In this paper, exploiting FDBP-OSD-0, we first explore the error-correcting performance of the 3D Chamon code, which is an instance of the 3D XYZ product code construction. For the isotropic 3D Chamon code, our simulation results show that the corresponding code-capacity thresholds under the depolarizing noise model, pure Pauli $X$, $Y$ and $Z$ noise model are all around $14.5\%$. These results are much higher than those in Ref. \cite{schwartzman2025generalizing} and \cite{zhao2023quantum}. Then we show that the 3D XYZ product can be generalized to four dimension and propose the 4D XYZ product code construction, which constructs a class of non-CSS codes by using four classical codes or two CSS codes. Compared with the 4D homological product, the 4D XYZ product can construct non-CSS codes with higher code dimension or code distance. 

To explore the error-correcting performance of 4D XYZ product codes, we study two instances—the 4D XYZ product concatenated code and the 4D Chamon code. Exploiting FDBP-OSD-0, under the pure Pauli $Z$ noise model, simulation results show that the code-capacity thresholds of the isotropic 4D Chamon code, the anisotropic 4D Chamon code and the 4D XYZ product concatenated code are around $13\%$, $19\%$, $38\%$ respectively, which are much higher than those of the isotropic 4D toric code, the anisotropic 4D toric code and the 4D homological product concatenated code, which are around $10.5\%$, $10.5\%$ and $15\%$ respectively. Besides, we also study their error-correcting performance under Pauli-$Z$-biased noise with bias rate $\eta=1000$, $100$, $10$ and $0.5$. For the anisotropic 4D Chamon code encoding 8 logical qubits, the corresponding code-capacity thresholds are around $18.5\%$, $18\%$, $17\%$ and $13\%$, respectively. For the 4D XYZ product concatenated code, the corresponding code-capacity thresholds are around $37\%$, $32\%$, $22\%$ and $10\%$, respectively.

These results indicate that, using two identical component CSS codes, the 4D XYZ product can construct non-CSS codes with superior error-correcting performance against Pauli-$Z$-biased noise compared with the CSS codes constructed by the 4D homological product.

Last but not least, we identify that the 4D Chamon code is an spin model with eight-qubit nearest-neighbor interactions on a 4D cubic lattice with periodic boundary conditions and demonstrate that it possesses two key characteristics of fracton models—ground state degeneracy dependent on cubic lattice size and the restricted mobility of excitations, which strongly suggests that it is a novel 4D fracton model. In addition, the error-correcting performance of the 4D Chamon code is significantly better than that of the other fracton model-the 4D X-cube model, which implies that the 4D Chamon code is a promising candidate for self-correcting quantum memory. More significantly, experimental realization of the 4D quantum surface code has recently been demonstrated in Quantinuum’s H2 trapped-ion quantum computer \cite{berthusen2024experiments}. The implementation exhibits superior error-correcting performance compared to the 2D surface code. This advancement substantiates the feasibility of implementing our proposed 4D Chamon code in practical quantum hardware.

\section*{Data Availability}
The data that support the findings of this study are available from the corresponding author upon reasonable request.

\section*{Acknowledgements}
This work is supported by the Colleges and Universities Stable Support Project of Shenzhen, China (No.GXWD20220817164856008), Shenzhen Science and Technology Program, China (JCYJ20241202123906009), Guangdong Provincial Key Laboratory of Novel Security Intelligence Technologies (2022B1212010005), the Colleges and Universities Stable Support Project of Shenzhen, China (No.GXWD20220811170225001) and Harbin Institute of Technology, Shenzhen - SpinQ quantum information Joint Research Center Project (No.HITSZ20230111).

\appendix
\section {CSS version of the 4D XYZ product}
\label{CSS version}
The non-CSS codes constructed by the 4D XYZ product can be transformed into CSS codes by finite-depth unitary circuits. In Eq. (\ref{4DXYZ stabilizer1}), for each qubit in parts $A$ and $E$, applying Clifford operation $U=SH$, where $H$ and $S$ are the Hadamard gate and phase gate, respectively. Similarly, for each qubit in parts $B$ and $D$, applying Clifford operation $U=HS$. We have,

\begin{widetext}
	\begin{equation}
		\label{4DXYZ stabilizer3}
		\mathcal{S}^\prime =\begin{bmatrix}
			S\\
			T\\
			U\\
			V
		\end{bmatrix} = \begin{bmatrix}
			Z^{\left(I_{m_1}\otimes{\widetilde{\delta}}_0^T\right)} &-Z^{\left(I_{m_1}\otimes{\widetilde{\delta}}_{-1}\right)} &Z^{\left(\delta_{-1}^T\otimes I_{n_B}\right)} &I &I\\
			X^{\left(\delta_{-1}\otimes I_{m_4}\right)} &I &X^{\left(I_{n_A}\otimes{\widetilde{\delta}}_0\right)} &X^{\left(\delta_0^T\otimes I_{m_4}\right)} &I\\
			I &X^{\left(\delta_{-1}\otimes I_{m_3}\right)} &X^{\left(I_{n_A}\otimes{\widetilde{\delta}}_{-1}^T\right)} &I &X^{\left(\delta_0^T\otimes I_{m_3}\right)}\\
			I &I &Z^{\left(\delta_0\otimes I_{n_B}\right)} &-Z^{\left(I_{m_2}\otimes{\widetilde{\delta}}_{0}^{T}\right)} &Z^{\left(I_{m_2}\otimes{\widetilde{\delta}}_{-1}\right)}
		\end{bmatrix}
	\end{equation}
\end{widetext}
The factor '$-1$' can be dropped, and we obtain a CSS code. One can check that all stabilizers in $\mathcal{S}^\prime$ commute. By this way, the tensor-product structure of $\mathcal{S}^\prime$ is shown in Fig. \ref{4DXYZProduct_CSS_Version}.

\begin{figure}[htbp]
	\centering
	\includegraphics[width=0.48\textwidth]{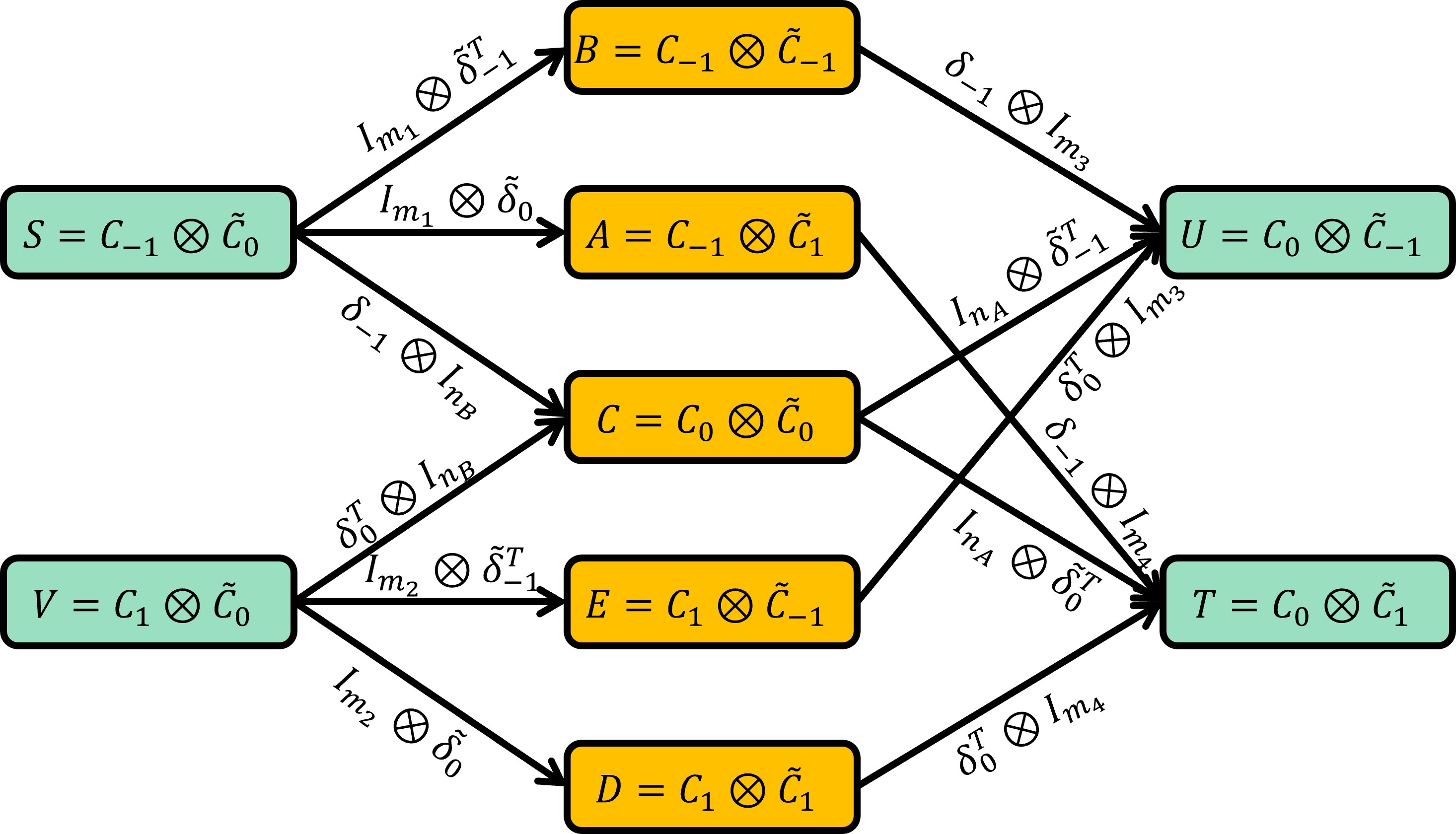}
	\caption{The tensor-product structure of CSS version of the 4D XYZ product of two length-2 chain complexes $\mathfrak{C}_1=C_{-1}\stackrel{\delta_{-1}}{\longrightarrow}C_0\stackrel{\delta_{0}}{\longrightarrow}C_1$ and $\mathfrak{C}_2={\widetilde{C}}_{-1}\stackrel{\widetilde{\delta}_{-1}}{\longrightarrow}{\widetilde{C}}_0\stackrel{\widetilde{\delta}_{0}}{\longrightarrow}{\widetilde{C}}_1$.}
	\label{4DXYZProduct_CSS_Version}
\end{figure}

Under this clifford deformation, the stabilizers of the deformed 4D Chamon code exhibit CSS structure: those in set $S$ (corresponding to cubes along $x$ or $y$ axis) and set $V$ (edges along $z$ or $w$ axis) are $Z$-type, whereas stabilizers in set $U$ (cubes along $z$ or $w$ axis) and set $T$ (edges along $x$- or $y$ axis) are $X$-type. 

Clifford deformation is an established method for tailoring quantum codes against biased noise \cite{bonilla2021xzzx,huang2023tailoring,setiawan2024tailoring}. For example, applying Hadamard gates to specific qubits in the rotated CSS surface codes yields the XZZX surface code \cite{bonilla2021xzzx}, which exhibits remarkable error-correcting performance for biased noise. Thus, this provides a novel perspective for understanding why the 4D XYZ product can construct codes with good error-correcting performance against biased noise.

\section {Code concatenation}
\label{Code concatenation}
Code concatenation is a procedure for constructing a larger quantum code from two smaller quantum codes. However, in our work, what we concatenate are two classical linear codes (specifically, classical repetition codes). This is because we utilize one repetition code as the bit-flip code which can only correct Pauli $X$ errors and the other as the phase-flip code which can only correct Pauli $Z$ errors. Their parity-check matrices both correspond to the classical repetition code. The well-known nine-qubit Shor code  \cite{PhysRevA.52.R2493} is a canonical example of the concatenated code, constructed by concatenating the three-qubit bit-flip code with the three-qubit phase-flip code.

Formally, consider two quantum codes, $C_1$ (the outer code) and $C_2$ (the inner code). Concatenating $C_1$ with $C_2$ involves replacing each physical qubit of $C_1$ with the logical qubit encoded by $C_2$. Consequently, every Pauli $P$ (where $P\in\{X,Y,Z\}$) operator in the stabilizers of $C_1$ must be replaced by the corresponding logical $P$ operator of $C_2$.

We use the nine-qubit Shor code as an illustrative example. Let $C_1$ be the three-qubit phase-flip code, stabilized by generators $\langle X_1X_2, X_2X_3 \rangle$ and $C_2$ be the three-qubit bit-flip code stabilized by generators $\langle Z_1 Z_2,Z_2 Z_3\rangle$ and with logical operator $\hat{X}_L=X_1X_2X_3$. Since $C_2$ serves as the inner code and encodes one logical qubit, replacing the three physical qubits of $C_1$ requires three copies of $C_2$. The resulting concatenated code—the nine-qubit Shor code—is stabilized by $X$-type stabilizers $\langle \hat{X}_L^{(1)}\hat{X}_L^{(2)}, \hat{X}_L^{(2)}\hat{X}_L^{(3)} \rangle = \langle X_1X_2X_3X_4X_5X_6,X_4X_5X_6X_7X_8X_9\rangle$, where the superscript $(i)$ (where $i\in\{1,2,3\}$) denotes the index of $C_2$ copy. Additionally, the $i^{th}$ copy of $C_2$ itself is stabilized by $\langle Z_1^{(i)} Z_2^{(i)},Z_2^{(i)} Z_3^{(i)}\rangle$. Therefore, the $Z$-type stabilizers of the nine-qubit Shor code are ${\langle Z_1^{(1)} Z_2^{(1)},Z_2^{(1)} Z_3^{(1)},Z_1^{(2)} Z_2^{(2)},Z_2^{(2)} Z_3^{(2)},Z_1^{(3)} Z_2^{(3)},Z_2^{(3)} Z_3^{(3)}\rangle} \allowbreak {= \langle Z_1Z_2,Z_2Z_3,Z_4Z_5,Z_5Z_6,Z_7Z_8,Z_8Z_9,\rangle}$.

The locality and LDPC properties of concatenated quantum codes depend critically on the choice of constituent codes. For the property of locality: concatenating repetition codes yields non-local stabilizer measurements (e.g., $X$-type stabilizers of the nine-qubit Shor code), whereas properly chosen constituent codes can achieve good locality—as demonstrated by the $XYZ^2$ hexagonal code in Ref. \cite{srivastava2025sequential}, which are constructed by concatenating a YZZY surface code with a $[[2,1,1]]$ phase-flip detecting code. Similarly for the property of LDPC: when replacing every Pauli $P$ (where $P\in\{X,Y,Z\}$) operator in $C_1$’s stabilizers by the corresponding logical $P$ operator of $C_2$, the resulting code cannot be LDPC if the weight of $C_2$’s logical operators scales with the code length. Conversely, if $C_2$’s logical operators have constant weight and $C_1$ is LDPC, the concatenated code will retain LDPC properties. This is exemplified by the $XYZ^2$ hexagonal code again.

\section {Proof of Theorem \ref{The logical operators of 4D XYZ product code}}
\label{Proof of Theorem 2}
\begin{widetext}
\begin{proof}
Our goal is to find out all independent logical operators, and our method is divided into three steps. First, finding out some operators that commutes with all stabilizer generators in Eq. (\ref{4DXYZ stabilizer2}). Second, excluding those operators which are stabilizers from the operators that we find in the first step, and the remaining operators are logical operators. Finally, we prove that the number of independent logical operator pairs that we find in the second step is equal to the code dimension of the corresponding 4D XYZ product code that we prove in \textbf{Theorem} \ref{The dimension of 4D XYZ product code}.
	
First, we consider the operators with the following form, which commute with all stabilizer generators,
	\begin{equation}
		\begin{aligned}
			&X_{G_1}=\left[I\ \ \ \ \ I\ \ \ \ X\left({\textbf{\emph{r}}^\prime}^\textbf{T}\right)\ \ \ I\ \ \ \ \ I\right]\\
			&Z_{G_1}=\left[I\ \ \ \ \ I\ \ \ \ X\left({\textbf{\emph{w}}^\prime}^\textbf{T}\right)\ \ \ I\ \ \ \ \ I\right]
		\end{aligned}		
	\end{equation}
	where
	\begin{equation}
		{\textbf{\emph{r}}^\prime}^T\in\ker\left(\begin{bmatrix}H_{x_1}\otimes I_{n_B}\\H_{z_1}\otimes I_{n_B}\end{bmatrix}\right)=\ker\left(\begin{bmatrix}H_{x_1}\\H_{z_1}\end{bmatrix}\right)\otimes\mathcal{C}_{n_B}
	\end{equation}
	and
	\begin{equation}
		{\textbf{\emph{w}}^\prime}^T\in\ker\left(\begin{bmatrix}I_{n_A}\otimes H_{x_2} \\I_{n_A}\otimes H_{z_2} \end{bmatrix}\right)=\mathcal{C}_{n_A}\otimes\ker\left(\begin{bmatrix}H_{x_1}\\H_{z_1}\end{bmatrix}\right)
	\end{equation}
	
	Let
	\begin{equation}
		{\textbf{\emph{r}}^{\prime\prime}}^T\in\ker\left(\begin{bmatrix}H_{x_1}\\H_{z_1}\end{bmatrix}\right)\otimes Im\left(H_{x_2}^T,H_{z_2}^T\right)\subset\ker\left(\begin{bmatrix}H_{x_1}\\H_{z_1}\end{bmatrix}\right)\otimes\mathcal{C}_{n_B}
	\end{equation}
	Next we prove that $X_{G_1^\prime}=\left[I\ I\ X\left({\textbf{\emph{r}}^{\prime\prime}}^{T}\right)\ I\ I\right]$ is a stabilizer which can only be generated by the product of some stabilizer generators in $T$ and $U$ and thus should be excluded. Considering an operator $P_1$ which is the product of some stabilizer generators in $T$ and $U$, it must have the following form 

	\begin{equation}
		\label{P1}
		\begin{aligned}
			P_1 &= \left[\textbf{\emph{t}}^T,\textbf{\emph{u}}^T\right]\begin{bmatrix}
				Y^{\left(H_{z_1}^T\otimes I_{m_4}\right)} &I &X^{\left(I_{n_A}\otimes H_{x_2}\right)} &Z^{\left(H_{x_1}^T\otimes I_{m_4}\right)} &I\\
				I &Z^{\left(H_{z_1}^T\otimes I_{m_3}\right)} &X^{\left(I_{n_A}\otimes H_{z_2}\right)} &I, &Y^{\left(H_{x_1}^T\otimes I_{m_3}\right)}
			\end{bmatrix}\\
			&=\left[Y^{\textbf{\emph{t}}^T\left(H_{z_1}^T\otimes I_{m_4}\right)}\ \ \ \ Z^{\textbf{\emph{u}}^T\left(H_{z_1}^T\otimes I_{m_3}\right)}\ \ \ \ X^{\textbf{\emph{t}}^T\left(I_{n_A}\otimes H_{x_2}\right)+\textbf{\emph{u}}^T\left(I_{n_A}\otimes H_{z_2}\right)}\ \ \ \ Z^{\textbf{\emph{t}}^T\left(H_{x_1}^T\otimes I_{m_4}\right)}\ \ \ \ Y^{\textbf{\emph{t}}^T\left(H_{x_1}^T\otimes I_{m_3}\right)}\right]
		\end{aligned}			
	\end{equation}
If the form of $P_1$ is the same with that of $X_{G_1^\prime}$, we must have
	\begin{equation}
		\begin{aligned}
			&H_{z_1}\otimes I_{m_4}\textbf{\emph{t}}=H_{x1}\otimes I_{m_4}\textbf{\emph{t}}=\mathbf{0}\\
			&H_{z1}\otimes I_{m_3}\textbf{\emph{u}}=H_{x1}\otimes I_{m_3}\textbf{\emph{u}}=\mathbf{0}
		\end{aligned}
	\end{equation}
	Thus, we have $\textbf{\emph{t}}\in\ker\left(\begin{bmatrix}H_{x_1}\\H_{z_1}\end{bmatrix}\right)\otimes\mathcal{C}_{m_4}$ and $\textbf{\emph{u}}\in\ker\left(\begin{bmatrix}H_{x_1}\\H_{z_1}\end{bmatrix}\right)\otimes\mathcal{C}_{m_3}$. Let $\textbf{\emph{t}}=\textbf{\emph{x}}\otimes \textbf{\emph{i}}$, and $\textbf{\emph{u}}=\textbf{\emph{y}}\otimes \textbf{\emph{j}}$, where $\textbf{\emph{x}},\textbf{\emph{y}}\in\ker\left(\begin{bmatrix}H_{x_1}\\H_{z_1}\end{bmatrix}\right)$, $\textbf{\emph{i}}\in\mathcal{C}_{m_4}$ and $\textbf{\emph{j}}\in\mathcal{C}_{m_3}$, we have
	\begin{equation}
		\begin{aligned}
			&\left(I_{n_A}\otimes H_{x_2}^T\right)\left(\textbf{\emph{x}}\otimes\textbf{\emph{i}}\right)+\left(I_{n_A}\otimes H_{z_2}^T\right)\left(\textbf{\emph{y}}\otimes\textbf{\emph{j}}\right)=\textbf{\emph{x}}\otimes H_{x_2}^T\textbf{\emph{i}}+\textbf{\emph{y}}\otimes H_{z_2}^T\textbf{\emph{j}}\in\ker\left(\begin{bmatrix}H_{x_1}\\H_{z_1}\end{bmatrix}\right)\otimes Im\left(H_{x_2}^T,H_{z_2}^T\right)
		\end{aligned}
	\end{equation}
	Thus, for vector $\textbf{\emph{r}}\in\ker\left(\begin{bmatrix}H_{x_1}\\H_{z_1}\end{bmatrix}\right)\otimes\left(\mathcal{C}_{n_B}\backslash Im\left(H_{x_2}^T,H_{z_2}^T\right)\right)$, the corresponding operator $X_L=\left(I\ I\ X\left(\textbf{\emph{r}}^T\right)\ I\ I\right)$ is a logical operator.
	
	Similarly, for vector $\textbf{\emph{w}}^{\prime\prime}\in Im\left(\left[H_{x1}^T,H_{z1}^T\right]\right)\otimes\ker\left(\begin{bmatrix}H_{x_2}\\H_{z_2}\end{bmatrix}\right)$,  the corresponding $Z_{G_1^\prime}=\left(I\ I\ Z\left({\textbf{\emph{w}}^{\prime\prime}}^{T}\right)\ I\ I\right)$ is a stabilizer which can only be generated by the product of some stabilizer generators in $S$ and $V$ and thus should be excluded.
	
	To sum up, for vectors 
	\begin{equation}
		\textbf{\emph{r}}\in \ker\left(\begin{bmatrix}H_{x_1}\\H_{z_1}\end{bmatrix}\right)\otimes \left(\mathcal{C}_{n_B}\backslash Im\left(\left[H_{x_2}^T,H_{z_2}^T\right]\right)\right)
	\end{equation}
	and
	\begin{equation}
		\textbf{\emph{w}}\in\left(\mathcal{C}_{n_A}\backslash Im\left(\left[H_{x_1}^T,H_{z_1}^T\right]\right)\right) \otimes \ker\left(\begin{bmatrix}H_{x_2}\\H_{z_2}\end{bmatrix}\right)
	\end{equation}
	the corresponding operators $X_L=\left(I\ I\ X\left(\textbf{\emph{r}}^T\right)\ I\ I\right)$ and $Z_L=\left(I\ I\ Z\left(\textbf{\emph{w}}^T\right)\ I\ I\right)$ must be logical operators.
	
	Notice that the dimension of  vector space $\mathcal{C}_{n_B}\backslash Im\left(\left[H_{x_2}^T,H_{z_2}^T\right]\right)$ is $n_B-\dim\left(Im\left(\left[H_{x_2}^T,H_{z_2}^T\right]\right)\right)= n_B-\dim\left(row\left(\begin{bmatrix}H_{x_2}\\H_{z_2}\end{bmatrix}\right)\right)=\dim\left(\ker\left(\begin{bmatrix}H_{x_2}\\H_{z_2}\end{bmatrix}\right)\right)$. Similarly, the dimension of vector space $\mathcal{C}_{n_A}\backslash Im\left(\left[H_{x_1}^T,H_{z_1}^T\right]\right)$ is equal to $n_A-\dim\left(row\left(\begin{bmatrix}H_{x_1}\\H_{z_1}\end{bmatrix}\right)\right)=\dim\left(\ker\left(\begin{bmatrix}H_{x_1}\\H_{z_1}\end{bmatrix}\right)\right)$. Thus, the numbers of independent logical operators $X_L$ and $Z_L$ are both $\dim\left(\ker\left(\begin{bmatrix}H_{x_1}\\H_{z_1}\end{bmatrix}\right)\right)\times \dim\left(\ker\left(\begin{bmatrix}H_{x_2}\\H_{z_2}\end{bmatrix}\right)\right)$. Moreover, for any vector $\textbf{\emph{r}}\in \ker\left(\begin{bmatrix}H_{x_1}\\H_{z_1}\end{bmatrix}\right)\otimes \left(\mathcal{C}_{n_B}\backslash Im\left(\left[H_{x_2}^T,H_{z_2}^T\right]\right)\right)$, we can find a vector $
	\textbf{\emph{w}}\in\left(\mathcal{C}_{n_A}\backslash Im\left(\left[H_{x_1}^T,H_{z_1}^T\right]\right)\right) \otimes \ker\left(\begin{bmatrix}H_{x_2}\\H_{z_2}\end{bmatrix}\right)$, such that $\textbf{\emph{r}}\cdot\textbf{\emph{w}}^T=1$, which means the corresponding $X_L$ and $Z_L$ anti-commute. Thus, there are $\dim\left(\ker\left(\begin{bmatrix}H_{x_1}\\H_{z_1}\end{bmatrix}\right)\right)\times \dim\left(\ker\left(\begin{bmatrix}H_{x_2}\\H_{z_2}\end{bmatrix}\right)\right)$ pairs of the first type logical operators.
	
	Second, we consider operators with the following form, which also commute with all stabilizer generators,
	\begin{equation}
		\begin{bmatrix}X_{G_2}\\X_{G_3}\\Z_{G_2}\\Z_{G_3}\end{bmatrix}=\begin{bmatrix}X({\textbf{\emph{a}}_1^\prime}^T)\ &I\ &I\ &Y({\textbf{\emph{b}}_1^\prime}^T)\ &I\\I\ &Y({\textbf{\emph{a}}_2^\prime}^T)\ &I\ &I\ &X({\textbf{\emph{b}}_2^\prime}^T)\\Y({\textbf{\emph{c}}_1^\prime}^T)\ &Z({\textbf{\emph{d}}_1^\prime}^T)\ &I\ &I\ &I\\I\ &I\ &I\ &Z({\textbf{\emph{c}}_2^\prime}^T)\ &Y({\textbf{\emph{d}}_2^\prime}^T)\end{bmatrix}
	\end{equation}
	where 
	\begin{equation}
		\left({\textbf{\emph{a}}_1^\prime}^T,\ {\textbf{\emph{b}}_1^\prime}^T\right)^T\in \ker{\left(\left[H_{z_1}^T\otimes I_{m_4},H_{x_1}^T\otimes I_{m_4}\right]\right)}=\ker{\left(\left[H_{z1}^T,H_{x1}^T\right]\right)}\otimes\mathcal{C}_{m_4}
	\end{equation}
	and 
	\begin{equation}
		\left({\textbf{\emph{a}}_2^\prime}^T,\ {\textbf{\emph{b}}_2^\prime}^T\right)^T\in \ker{\left(\left[H_{z_1}^T\otimes I_{m_3},H_{x_1}^T\otimes I_{m_3}\right]\right)}=\ker{\left(\left[H_{z1}^T,H_{x1}^T\right]\right)}\otimes\mathcal{C}_{m_3}
	\end{equation}
	
	Let
	\begin{equation}
		{\pi\left({\textbf{\emph{a}}_1^{\prime\prime}}^T,\ {\textbf{\emph{a}}_2^{\prime\prime}}^T,{\textbf{\emph{b}}_1^{\prime\prime}}^T,{\textbf{\emph{b}}_2^{\prime\prime}}^T\right)}^T\in \ker{\left(\left[H_{z1}^T,H_{x1}^T\right]\right)}\otimes Im\left(\begin{bmatrix}H_{x_2}\\ H_{z_2}\end{bmatrix}\right)
	\end{equation}
	and
	\begin{equation}
	\begin{bmatrix}X_{G_2}^{\prime}\\X_{G_3}^{\prime}\end{bmatrix}=\begin{bmatrix}X({\textbf{\emph{a}}_1^{\prime\prime}}^T)\ &I\ &I\ &Y({\textbf{\emph{b}}_1^{\prime\prime}}^T)\ &I\\I\ &Y({\textbf{\emph{a}}_2^{\prime\prime}}^T)\ &I\ &I\ &X({\textbf{\emph{b}}_2^{\prime\prime}}^T)\end{bmatrix}
	\end{equation}
	Next, we prove that $X_{G_2^{\prime}}X_{G_3^{\prime}}=\left[X\left({\textbf{\emph{a}}_1^{\prime\prime}}^T\right)\ Y\left({\textbf{\emph{a}}_2^{\prime\prime}}^T\right)\ I\ Y\left({\textbf{\emph{b}}_1^{\prime\prime}}^T\right)\ X\left({\textbf{\emph{b}}_2^{\prime\prime}}^T\right)\right]$ is a stabilizer, which can only be generated by the product of some stabilizer generators in $S$ and $V$, thus should be excluded. Considering an operator $P_2$ which is the product of some stabilizer generators in $S$ and $V$, it must have the following form
	\begin{equation}
		\label{P2}
		\begin{aligned}
			P_2 &= \left[\textbf{\emph{s}}^T,\textbf{\emph{v}}^T\right]\begin{bmatrix}
				X^{\left(I_{m_1}\otimes H_{x_2}^T\right)} &Y^{\left(I_{m_1}\otimes H_{z_2}^T\right)} &Z^{\left(H_{z_1}\otimes I_{n_B}\right)} &I &I\\
				I &I &Z^{\left(H_{x1}\otimes I_{n_B}\right)} &Y^{\left(I_{m_2}\otimes H_{x_2}^T\right)} &X^{\left(I_{m_2}\otimes H_{z_2}^T\right)}
			\end{bmatrix}\\
			&=\left[X^{\textbf{\emph{s}}^T\left(I_{m_1}\otimes H_{x_2}^T\right)}\ \ \ \ \ Y^{\textbf{\emph{s}}^T\left(I_{m_1}\otimes H_{z_2}^T\right)}\ \ \ \ Z^{\textbf{\emph{s}}^T\left(H_{z_1}\otimes I_{n_B}\right)+\textbf{\emph{v}}^T\left(H_{x1}\otimes I_{n_B}\right)}\ \ \ \ Y^{\textbf{\emph{v}}^T\left(I_{m_2}\otimes H_{x_2}^T\right)}\ \ \ \ X^{\textbf{\emph{v}}^T\left(I_{m_2}\otimes H_{z_2}^T\right)}\right]
		\end{aligned}			
	\end{equation}
	If the form of $P_2$ is the same with $X_{G_2^{\prime}}X_{G_3^{\prime}}$, we must have
	\begin{equation}
		H_{z1}^T\otimes I_{n_B}\textbf{\emph{s}}+H_{x1}^T\otimes I_{n_B}\textbf{\emph{v}}=\mathbf{0}
	\end{equation}
	Thus, we have $\begin{bmatrix}\textbf{\emph{s}}\\ \textbf{\emph{v}}\end{bmatrix}\in \ker\left(\left[H_{z1}^T\otimes I_{n_B},H_{x1}^T\otimes I_{n_B}\right]\right)=\ker \left(\left[H_{z_1}^T, H_{x_1}^T\right]\right)\otimes \mathcal{C}_{n_B}$. Let $\begin{bmatrix}\textbf{\emph{s}}\\ \textbf{\emph{v}}\end{bmatrix}=\begin{bmatrix}\textbf{\emph{x}}\\ \textbf{\emph{y}}\end{bmatrix}\otimes\textbf{\emph{i}}$, where $\begin{bmatrix}\textbf{\emph{x}}\\ \textbf{\emph{y}}\end{bmatrix}=\left[x_1,\cdots,x_{m_1},y_1,\cdots,y_{m_2}\right]^T\in\ker{\left(\left[H_{z_1}^T,H_{x_1}^T\right]\right)}$ and $\textbf{\emph{i}}\in\mathcal{C}_{n_B}$, we have
	\begin{equation}
		\begin{aligned}
			\begin{bmatrix}
				I_{m_1}\otimes H_{x_2} &\textbf{0}\\
				I_{m_1}\otimes H_{z_2} &\textbf{0}\\
				\textbf{0} &I_{m_2}\otimes H_{x_2}\\
				\textbf{0} &I_{m_2}\otimes H_{z_2}
			\end{bmatrix}\begin{bmatrix}\textbf{\emph{x}}\\ \textbf{\emph{y}}\end{bmatrix}\otimes\textbf{\emph{i}}=\begin{bmatrix}\textbf{\emph{x}}\otimes H_{x_2}\textbf{\emph{i}}\\\textbf{\emph{x}}\otimes H_{z_2}\textbf{\emph{i}}\\\textbf{\emph{y}}\otimes H_{x_2}\textbf{\emph{i}}\\\textbf{\emph{y}}\otimes H_{z_2}\textbf{\emph{i}}\end{bmatrix}=\begin{bmatrix}x_1H_{x_2}\textbf{\emph{i}}\\ \vdots\\ x_{m_1}H_{x_2}\textbf{\emph{i}}\\ x_1H_{z_2}\textbf{\emph{i}}\\ \vdots\\ x_{m_1}H_{z_2}\textbf{\emph{i}}\\ y_1H_{x_2}\textbf{\emph{i}}\\ \vdots\\ y_{m_2}H_{x_2}\textbf{\emph{i}}\\ y_1H_{z_2}\textbf{\emph{i}}\\ \vdots\\ y_{m_2}H_{z_2}\textbf{\emph{i}}\end{bmatrix}
		\end{aligned}
	\end{equation}
	
	It can be seen that
	\begin{equation}
		\begin{aligned}
			\pi\left(\begin{bmatrix}x_1H_{x_2}\textbf{\emph{i}}\\ \vdots\\ x_{m_1}H_{x_2}\textbf{\emph{i}}\\ x_1H_{z_2}\textbf{\emph{i}}\\ \vdots\\ x_{m_1}H_{z_2}\textbf{\emph{i}}\\ y_1H_{x_2}\textbf{\emph{i}}\\ \vdots\\ y_{m_2}H_{x_2}\textbf{\emph{i}}\\ y_1H_{z_2}\textbf{\emph{i}}\\ \vdots\\ y_{m_2}H_{z_2}\textbf{\emph{i}}\end{bmatrix}\right)=\begin{bmatrix}x_1H_{x_2}\textbf{\emph{i}}\\ x_1H_{z_2}\textbf{\emph{i}}\\ \vdots\\ x_{m_1}H_{x_2}\textbf{\emph{i}}\\x_{m_1}H_{z_2}\textbf{\emph{i}}\\ y_1H_{x_2}\textbf{\emph{i}}\\ y_1H_{z_2}\textbf{\emph{i}}\\ \vdots\\ y_{m_2}H_{x_2}\textbf{\emph{i}}\\y_{m_2}H_{z_2}\textbf{\emph{i}}\end{bmatrix}\in\ker{\left(\left[H_{z_1}^T,H_{x_1}^T\right]\right)}\otimes Im\begin{bmatrix}
				H_{x_2}\\H_{z_2}
			\end{bmatrix}
		\end{aligned}		
	\end{equation}
	which means $X_{G_2^{\prime}}X_{G_3^{\prime}}$ is a stabilizer. Thus, for vectors
	\begin{equation}
		{\pi\left(\textbf{\emph{a}}_1^T,\textbf{\emph{a}}_2^T,\textbf{\emph{b}}_1^T,\textbf{\emph{b}}_2^T\right)}^T\in\ker{\left(\left[H_{z_1}^T,H_{x_1}^T\right]\right)}\otimes \left(\mathcal{C}_{m_3+m_4}\backslash Im\begin{bmatrix}
			H_{x_2}\\H_{z_2}
		\end{bmatrix}\right)
	\end{equation}
	the corresponding operators $\begin{bmatrix}X_{L_2}\\X_{L_3}\end{bmatrix}=\begin{bmatrix}X({\textbf{\emph{a}}_1}^T)\ &I\ &I\ &Y({\textbf{\emph{b}}_1}^T)\ &I\\I\ &Y({\textbf{\emph{a}}_2}^T)\ &I\ &I\ &X({\textbf{\emph{b}}_2}^T)\end{bmatrix}$ are logical operators.
	
	Similarly, for vector
	\begin{equation}
		\left({\textbf{\emph{c}}_1^{\prime\prime}}^T,{\textbf{\emph{d}}_1^{\prime\prime}}^T,{\textbf{\emph{c}}_2^{\prime\prime}}^T,{\textbf{\emph{d}}_2^{\prime\prime}}^T\right)^T\in Im\left(\begin{bmatrix}H_{z_1}\\H_{x_1}\end{bmatrix}\right)\otimes\ker\left(\left[H_{x_2}^T,H_{z_2}^T\right]\right)
	\end{equation}
	and
	\begin{equation}
	\begin{bmatrix}Z_{G_2}^{\prime}\\Z_{G_3}^{\prime}\end{bmatrix}=\begin{bmatrix}Y({\textbf{\emph{c}}_1^{\prime\prime}}^T)\ &Z({\textbf{\emph{d}}_1^{\prime\prime}}^T)\ &I\ &I\ &I\\I\ &I\ &I\ &Z({\textbf{\emph{c}}_2^{\prime\prime}}^T)\ &Y({\textbf{\emph{d}}_2^{\prime\prime}}^T)\end{bmatrix}
	\end{equation}
	the operator $Z_{G_2^{\prime}}Z_{G_3^{\prime}}=\left[Y\left({\textbf{\emph{c}}_1^{\prime\prime}}^T\right)\ \ Z\left({\textbf{\emph{d}}_1^{\prime\prime}}^T\right)\ \ I\ \ Z\left({\textbf{\emph{c}}_2^{\prime\prime}}^T\right)\ \ Y\left({\textbf{\emph{d}}_2^{\prime\prime}}^T\right)\right]$ is a stabilizer which can only be generated by the product of some stabilizer generators in $T$ and $U$ and thus should be excluded.
	
	To sum up, for vectors 
	\begin{equation}
		{\pi\left(\textbf{\emph{a}}_1^T,\ \textbf{\emph{a}}_2^T,\textbf{\emph{b}}_1^T,\textbf{\emph{b}}_2^T\right)}^T\in \ker\left(\left[H_{z_1}^T,H_{x_1}^T\right]\right)\otimes\left(\mathcal{C}_{m_3+m_4}\backslash Im\left(\begin{bmatrix}H_{x_2}\\H_{z_2}\end{bmatrix}\right)\right)
	\end{equation}
	and 
	\begin{equation}
		{\pi\left(\textbf{\emph{c}}_1^T,\ \textbf{\emph{d}}_1^T,\textbf{\emph{c}}_2^T,\textbf{\emph{d}}_2^T\right)}^T\in\left(\mathcal{C}_{m_1+m_2}\backslash Im\left(\begin{bmatrix}H_{z_1}\\H_{x_1}\end{bmatrix}\right)\right)\otimes\ker\left(\left[H_{x_2}^T,H_{z_2}^T\right]\right)
	\end{equation}
	the corresponding $\begin{bmatrix}X_{L_2}\\X_{L_3}\end{bmatrix}=\begin{bmatrix}X(\textbf{\emph{a}}_1)\ &I\ &I\ &Y(\textbf{\emph{b}}_1)\ &I\\I\ &Y(\textbf{\emph{a}}_2)\ &I\ &I\ &X(\textbf{\emph{b}}_2)\end{bmatrix}$ and $\begin{bmatrix}Z_{L_2}\\Z_{L_3}\end{bmatrix}=\begin{bmatrix}Y(\textbf{\emph{c}}_1)\ &Z(\textbf{\emph{d}}_1)\ &I\ &I\ &I\\I\ &I\ &I\ &Z(\textbf{\emph{c}}_2)\ &Y(\textbf{\emph{d}}_2)\end{bmatrix}$ must be logical operators.
	
	Notice that the dimension of vector space $\mathcal{C}_{m_3+m_4}\backslash Im\begin{bmatrix}
		H_{x_2}\\H_{z_2}
	\end{bmatrix}$ is $m_3+m_4-\dim\left(Im\begin{bmatrix}
		H_{x_2}\\H_{z_2}
	\end{bmatrix}\right)=m_3+m_4-\dim\left(row\left(\left[H_{x_2}^T,H_{z_2}^T\right]\right)\right)=\dim\left(\ker\left(\left[H_{x_2}^T,H_{z_2}^T\right]\right)\right)$. Similarly, the dimension of vector space $\mathcal{C}_{m_1+m_1}\backslash Im\begin{bmatrix}
		H_{z_1}\\H_{x_1}
	\end{bmatrix}$ is $m_1+m_2-\dim\left(Im\begin{bmatrix}
		H_{z_1}\\H_{x_1}
	\end{bmatrix}\right)=m_1+m_2-\dim\left(row\left(\left[H_{z_1}^T,H_{x_1}^T\right]\right)\right)=\dim\left(\ker\left(\left[H_{z_1}^T,H_{x_1}^T\right]\right)\right)$. Thus, the total number of independent logical operators $X_{L_2}$ and $X_{L_3}$ and that of $Z_{L_2}$ and $Z_{L_3}$ are both $\dim\left(\ker{\left(\left[H_{z_1}^T,H_{x_1}^T\right]\right)}\right)\times \dim\left(\ker\left(\left[H_{x_2}^T,H_{z_2}^T\right]\right)\right)$. Moreover, for any vector $\textbf{\emph{r}}\in \ker\left(\left[H_{z_1}^T,H_{x_1}^T\right]\right)\otimes\left(\mathcal{C}_{m_3+m_4}\backslash Im\left(\begin{bmatrix}H_{x_2}\\H_{z_2}\end{bmatrix}\right)\right)$, we can find a vector $\textbf{\emph{w}}\in\left(\mathcal{C}_{m_1+m_2}\backslash Im\left(\begin{bmatrix}H_{x_1}\\H_{z_1}\end{bmatrix}\right)\right)\otimes\ker\left(\left[H_{z_2}^T,H_{x_2}^T\right]\right)$, such that $\textbf{\emph{r}}\cdot\textbf{\emph{w}}^T=1$, which means the corresponding $X_L$ and $Z_L$ anti-commute. Thus, there are $\dim\left(\ker{\left(\left[H_{z_1}^T,H_{x_1}^T\right]\right)}\right)\times \dim\left(\ker{\left(\left[H_{z_2}^T,H_{x_2}^T\right]\right)}\right)$ pairs of the second type logical operators.
	
	The total number of the first and the second types of logical operators is $\dim\left(\ker\left(\begin{bmatrix}H_{x_1}\\H_{z_1}\end{bmatrix}\right)\right)\times \dim\left(\ker\left(\begin{bmatrix}H_{x_2}\\H_{z_2}\end{bmatrix}\right)\right)+\dim\left(\ker{\left(\left[H_{z_1}^T,H_{x_1}^T\right]\right)}\right)\times \dim\left(\ker{\left(\left[H_{z_2}^T,H_{x_2}^T\right]\right)}\right)$. Now, we prove that this number is equal to the code dimension of the corresponding 4D XYZ product code as we prove in \textbf{Theorem} \ref{The dimension of 4D XYZ product code}.
	
	Let $a=\dim\left(\ker\left(\begin{bmatrix}H_{x_1}\\H_{z_1}\end{bmatrix}\right)\right)$ and $b=\dim\left(\ker\left(\begin{bmatrix}H_{x_2}\\H_{z_2}\end{bmatrix}\right)\right)$, we have $\dim\left(\ker{\left(\left[H_{z_1}^T,H_{x_1}^T\right]\right)}\right)=m_1+m_2-\dim\left(row{\left(\left[H_{z_1}^T,H_{x_1}^T\right]\right)}\right)=m_1+m_2-\left(n_A-\dim\left(\ker\left(\begin{bmatrix}H_{x_1}\\H_{z_1}\end{bmatrix}\right)\right)\right)=m_1+m_2-(n_A-a)$. Similarly, $\dim\left(\ker{\left(\left[H_{z_2}^T,H_{x_2}^T\right]\right)}\right)=m_3+m_4-\dim\left(row{\left(\left[H_{z_2}^T,H_{x_2}^T\right]\right)}\right)=m_3+m_4-\left(n_B-\dim\left(\ker\left(\begin{bmatrix}H_{x_2}\\H_{z_2}\end{bmatrix}\right)\right)\right)=m_3+m_4-(n_B-b)$. Thus,

	\begin{equation}
		\begin{aligned}
			&\dim\left(\ker\left(\begin{bmatrix}H_{x_1}\\H_{z_1}\end{bmatrix}\right)\right)\times \dim\left(\ker\left(\begin{bmatrix}H_{x_2}\\H_{z_2}\end{bmatrix}\right)\right)+\dim\left(\ker{\left(\left[H_{z_1}^T,H_{x_1}^T\right]\right)}\right)\times \dim\left(\ker{\left(\left[H_{z_2}^T,H_{x_2}^T\right]\right)}\right)\\
			&=ab+\left[m_1+m_2-\left(n_A-a\right)\right]\left[m_3+m_4-\left(n_B-b\right)\right]\\
			&=a\left[b-\left(n_B-m_3-m_4\right)\right]+b\left[a-\left(n_A-m_1-m_2\right)\right]+\left(n_A-m_1-m_2\right)\left(n_B-m_3-m_4\right)
		\end{aligned}		
	\end{equation}
which is equal to the code dimension as we prove in \textbf{Theorem} \ref{The dimension of 4D XYZ product code}. Thus, the proof is completed.
\end{proof}
\end{widetext}

\section {Proof of Lemma \ref{minimum weight}}
\label{Proof of Lemma 2}
\begin{proof}
	Here, we prove the minimum weight of non-zero vectors in space $\mathcal{C}_{n}\backslash Im (H_1)$ with that of $\mathcal{C}_{n}\backslash \ker (H_2)$ following a similar fashion.
	
	The space $\mathcal{C}_{n}$ can be represented by the row space of an $n\times n$ identity matrix $I_n$, namely, $row(I_n)$. Thus, this problem is equivalent to finding the minimum weight of non-zero vectors in space $row(I_n)$ but not in $Im (H_1)$. To do this, one can construct a block matrix $\left[H_1, I_n\right]$, then using Gaussian elimination to find out a set of column indices of its pivots $\textbf{P}$. For the rows which are in $I_n$ but not in $H_1^T$ and whose indices are in $\textbf{P}$, they are the vectors in space $row(I_n)$ but not in $Im (H_1)$. Since the weight of each row of $I_n$ is one, thus the minimum weight of non-zero vectors in space $\mathcal{C}_{n}\backslash Im (H_1)$ is one and the proof is completed.
\end{proof}

\section {Comparison of number of 4-cycles}
\label{number of 4-cycles}
Table \ref{4-cycles} gives the number of 4-cycles in the Tanner graphs of the 3D Chamon code, 4D Chamon code, 3D toric code and 4D toric code. One can see that the numbers of 4-cycles in the Tanner graph of the 3D and 4D Chamon codes are much higher than those of the 3D and 4D toric codes.

\begin{table}[htbp]
	\begin{center}
		\caption{Comparison of the 3D Chamon code, 4D Chamon code, 3D toric code and 4D toric code in the number of 4-cycles.}		
		\begin{tabular}{c|c|c}
			\hline
			& $n_1$, $n_2$, $n_3$ &  number of 4-cycles \\
			\hline
			\multirow{5}{*}{3D Chamon code} & $2,2,2$ & 240 \\
			\cline{2-3}
			& $3,3,3$ & 648\\
			\cline{2-3}
			& $4,4,4$ & 1536\\
			\cline{2-3}
			& $2,3,4$ & 624\\
			\cline{2-3}	
			& $3,4,5$ & 1440\\		
			\hline
			\multirow{5}{*}{3D toric code} & $2,2,2$ & 132\\
			\cline{2-3}
			& $3,3,3$ & 324\\
			\cline{2-3}
			& $4,4,4$ & 768\\
			\cline{2-3}
			& $2,3,4$ & 201\\
			\cline{2-3}	
			& $3,4,5$ & 467\\		
			\hline
			& $n_1$, $n_2$, $n_3$, $n_4$ &  number of 4-cycles \\
			\hline
			\multirow{2}{*}{4D Chamon code} & $2,2,2,2$ & 1792 \\
			\cline{2-3}
			& $2,3,2,3$ & 3744\\
			\hline
			\multirow{2}{*}{4D toric code} & $2,2,2,2$ & 519 \\
			\cline{2-3}
			& $2,3,2,3$ & 1095\\
			\hline
		\end{tabular}
		\label{4-cycles}
	\end{center}
\end{table}


\newpage
\bibliography{sn-bibliography}

\end{document}